%% file: Main_NDSS2024.tex
\documentclass[conference]{IEEEtran}
\ifCLASSINFOpdf
\else
\fi
\let\labelindent\relax
\usepackage{tabularx}
\usepackage{array}
\usepackage{arydshln}
\usepackage{amsfonts}
\usepackage{amssymb}
\usepackage{amsmath}
\usepackage{amsthm}
\usepackage{mathtools}
\usepackage{multirow}
\usepackage{algorithm}
\usepackage{algorithmic}
\usepackage{multirow}
\usepackage{xcolor}
\usepackage{enumerate}
\usepackage{enumitem}

\usepackage{url}
\usepackage{color}
\usepackage{tikz}
\usepackage{bbding}
\usepackage{framed} 
\usepackage{comment}
\usepackage{subfigure}
\usepackage{pifont}
\usepackage{graphicx}
\usepackage[justification=centering]{caption}
\usepackage[export]{adjustbox}

\usepackage[
n,
operators,
advantage,
sets,
adversary,
landau,
probability,
notions,	
logic,
ff,
mm,
primitives,
events,
complexity,
asymptotics,
keys]{cryptocode}
\usepackage{lipsum}
\usepackage{hyperref}

\DeclareMathAlphabet{\mathrmbf}{\encodingdefault}{\rmdefault}{bx}{n}
\hypersetup{colorlinks=true,linkcolor=black,citecolor=black,urlcolor=black}
\newcommand{\ignore}[1]{}

\newcommand{\tian}[1]{\textcolor{blue}{Tian:#1}}
\newcommand{\yanan}[1]{\colorbox{lightgray}{Yanan:}{\color{blue}#1}}

\newcommand{\red}[1]{\textcolor{red}{#1}}

\newtheorem{remark}{Remark}

\newtheorem{definition}{Definition}

\newtheorem{theorem}{Theorem}

\newcommand{\name}{Pisces}

\begin{document}
%

\title{$\mathsf{\name}$: Private and Compliable Cryptocurrency Exchange}


\author{\IEEEauthorblockN{Ya-Nan Li}
\IEEEauthorblockA{The University of Sydney\\
yanan.li@sydney.edu.au}
\and
\IEEEauthorblockN{Tian Qiu}
\IEEEauthorblockA{The University of Sydney\\
tqiu4893@uni.sydney.edu.au}
\and
\IEEEauthorblockN{Qiang Tang}
\IEEEauthorblockA{The University of Sydney\\
qiang.tang@sydney.edu.au}}


%


\IEEEoverridecommandlockouts
\makeatletter\def\@IEEEpubidpullup{6.5\baselineskip}\makeatother
\IEEEpubid{\parbox{\columnwidth}{
    Network and Distributed System Security (NDSS) Symposium 2024\\
    26 February - 1 March 2024, San Diego, CA, USA\\
    ISBN 1-891562-93-2\\
    https://dx.doi.org/10.14722/ndss.2024.23xxx\\
    www.ndss-symposium.org
}
\hspace{\columnsep}\makebox[\columnwidth]{}}

\maketitle

\begin{abstract}
Cryptocurrency exchange platforms such as Coinbase, Binance, enable users to purchase and sell cryptocurrencies conveniently just like trading stocks/commodities. However, because of the nature of blockchain, when a user withdraws coins (i.e., transfers coins to an external on-chain account), all future transactions can be learned by the platform. This is in sharp contrast to conventional stock exchange where all {\em external} activities of users are always hidden from the platform. Since the platform knows highly sensitive user private information such as passport number, bank information etc, linking all (on-chain) transactions raises a serious privacy concern about the potential disastrous data breach in those cryptocurrency exchange platforms.

In this paper, we propose a cryptocurrency exchange that restores user anonymity for the first time. To our surprise, the seemingly well-studied privacy/anonymity problem has several new challenges in this setting. Since the public blockchain and internal transaction activities naturally provide many non-trivial leakages to the platform, internal privacy is not only useful in the usual sense but also becomes necessary for regaining the basic anonymity of user transactions. We also ensure that the user cannot double spend, and the user has to properly report {\em accumulated} profit for tax purposes, even in the private setting.  We give a careful modeling and efficient construction of the system that achieves constant computation and communication overhead (with only simple cryptographic tools and rigorous security analysis); we also implement our system and evaluate its practical performance.
\end{abstract}


%


\section{Introduction}

Just like stocks and other commodities, people buy or sell cryptocurrencies on exchange platforms, mostly, on centralized platforms such as Coinbase, which are essentially marketplaces for cryptocurrencies. There, customers can pay fiat money like U.S dollars to get some coin, e.g, Bitcoin, or transfer their coin in the platform to an external account ({\em withdrawal}) \footnote{Or transfer coins into accounts in the platform from an external account ({\em deposit}), then sell for fiat money; and {\em exchange} one coin, e.g., BTC, to get some other coin, e.g., ETH. See Fig.\ref{overview}, and Sec.\ref{syntax} for details.}. 
Despite the promise of decentralized exchange, those centralized trading platforms still play a major role for usability and even regulatory reasons. For example, the annual trading volume of Binance was up to 9580 billion USD in 2021~\cite{BinanceStatistics}, and Coinbase also had 1640 billion USD in 2021 \cite{CoinbaseStatistics}. 

Like conventional stock exchanges, these exchange platforms must comply with regulations including Know Your Customer (KYC). They require businesses to verify the identity of their clients. Essentially, when a client/user registers an account at the exchange platform, he is normally required to provide a real-world identification document, such as passport, or a stamped envelope with address, for the platform to verify. Also for trading purposes, bank information is also given.

\noindent{\bf Serious privacy threats.} Despite provided convenience, those centralized cryptocurrency exchange platforms cause a much more serious concern on potential privacy breaches. 

As we have witnessed, many data breach instances exist \cite{DataBreaches}. A more worrisome issue in the exchange setting is that exchange can be seen as a bridge between the real world and the cryptocurrency world, which amplify the impacts of potential privacy breach (in exchange, of user records including identities and accounts). Users may {\em deposit} coins into the platform from,  or {\em withdraw} coins (transfer out of the platform) to,  his personal account on a blockchain. 


Since most of the blockchains are transparent (except very few number of chains such as Zcash \cite{zcash}) and publicly accessible (e.g., Bitcoin, Ethereum), the platform can essentially extract all transaction history knowing the real identity of a user. In the former case, the platform immediately links the real identity to his incoming addresses, and trace back all previous transactions on-chain; while even worse in the latter case, the platform, knowing the real identity of a user, and knows exactly which account/address of the cryptocurrency the user requested to withdraw (transfer to), and all future transactions. For instance, the platform could easily deduce that a user Alice bought a Tesla car with Bitcoin, as she withdrew them from the platform and then transferred them to Tesla's Bitcoin account (which could be public knowledge). 


\ignore{
It sells cryptocurrencies to users and holds their assets in escrow. 
Note that the user's account in the platform binds to his real identity due to the requirement of KYC and AML. 
The platform also knows their cryptocurrency addresses on the blockchain when they request to withdraw cryptocurrencies to the designated wallet addresses.

the platform knows users' future spending behaviors and destruct their privacy.
The blockchain just shows that ``somebody'' pays Tesla with the amount of Bitcoins, but Coinbase knows exactly that the payer is Alice.
}

It follows that existing centralized cryptocurrency exchange immediately ``destroys'' the pseudonym protection of blockchains, and the platform could obtain a large amount of information that is not supposed to be learned, e.g., the purchase/transaction histories of clients {\em outside of} the platform. This is even worse than conventional stock/commodity exchange where privacy may be breached within the system, but user information outside of the system is not revealed. 

We would like to design a cryptocurrency exchange system that at least restores user anonymity/privacy so that external records are not directly linked to the real identity. 




\noindent{\bf Insufficiency of external anonymity mechanisms.} 
The first potential method is keeping the existing exchange unchanged, and cutting the link between the exchange and external blockchain by
making on-chain payments/transfers (for coin deposits and withdrawals) on every blockchain anonymous, 
so that nobody 
can link the payer and the payee. Unfortunately, all those external anonymity solutions are insufficient.

First, fully anonymous on-chain payments such as Zcash only support its own native coins, while in most exchange platforms, Bitcoin, Ethereum and many other crypto tokens are the main objects of exchange, and cannot be supported. 


More exotic solutions like anonymous layer-2 payment solutions ~\cite{heilman_tumblebit_2017, glaeser2022foundations, qin2022blindhub, ng_ldsp_2021} and private smart contract enabled private payment solutions~\cite{bunz2020zether,diamond2021manyzether} also exist. 
One may wonder whether we can let the platform be the payer in those solutions during coin withdrawal. 
%
%
%
However, existing solutions mainly consider $k$-anonymity (where $k$ is the number of active users in an epoch) against the hub in~\cite{heilman_tumblebit_2017, glaeser2022foundations, qin2022blindhub} and the leader in~\cite{ng_ldsp_2021} and other outsiders, {\em not against the payer himself}. In our case, the platform is the payer and knows exactly the payee address during a coin withdrawal. 

Recent works of \cite{tairia2l2021,glaeser2022foundations} even considered anonymous ($k$-anonymity) payment hub against payers, assuming fixed denomination. Besides that $k$ is usually small, the withdraw transactions in exchange platform can hardly be of a fixed amount. When two users withdraw different amount of coins, the platform again can trivially tell them apart.

\noindent{\bf Unexplored anonymity within exchange platform.} The above analysis hints that relying on external anonymity mechanism alone is insufficient, we need to further strengthen the anonymity protection {\em within} the platform. Anonymity issues are classical topics that have been extensively studied in different settings, including in cryptocurrencies; yet, we will demonstrate that large body of those works are not applicable to our setting of exchange system.

First, not only anonymous payment hub solutions cannot be directly applicable, even the techniques
(e.g., viewing the exchange platform as the hub instead, while each user can be both a payer and payee) are not sufficient either for the {\em ``internal''} anonymity. The key difference, again, lies in the functionality difference of payment hubs (and other payment related solutions in general) and exchange platforms. 

Usually, in an anonymous payment hub, payer-payee exchange some information first, and then each runs some form of (blockchain facilitated) 
 fair exchange protocol with the hub. For anonymity, they would require a bunch of payers and payees to have some on-chain {\em setup} first with the hub, and $k$ active payments, so that the link between each pair of payer-payee can be hidden among those $k$ transactions; otherwise, each individual incoming transaction can be recognized by the hub. But in an exchange platform, there is no other entity for such setup, each individual request would be independent from the view of the platform: when user A, B purchase some BTCs from the exchange platform, 
 these purchase requests can trivially be distinguished by the platform (i.e., $k=1$).

 Another issue (not covered in the payment solutions) in anonymous exchange is that every exchange transaction between a user and the platform contains two highly correlated parts: the transaction from user to platform and that from platform to user. The amounts are based on the exchange rate, e.g., A pays 1 BTC, for 15 ETHs. While in (anonymous) payment solutions, any two transactions can be completely independent, e.g., A pays 10 BTCs to B (e.g. platform here), while B pays 1 ETH to A.

There are also some works on private Decentralized EXchange (DEX for short) \cite{bowe2020zexe, Baum2021P2DEXExchange} where users exchange cryptocurrencies with each other. 
The privacy model in DEX is different from that of our centralized setting.
It 
keeps the transaction information secret {\em except} for the trading parties.
Again, in our setting, the platform is one of the trading parties who can learn the information of the other trivially.


Atomic swap across different ledgers supports the exchange between different cryptocurrencies.
While atomic swap pays much effort on ensuring fairness, the only privacy-preserving atomic swap work \cite{deshpande2020privateswap} reduces
the confidentiality and anonymity properties to the underlying blockchains. 
If the swap protocol involves cryptocurrencies on transparent blockchains, like Bitcoin and Ethereum, these two transactions can be linked easily via their amounts. 
%
There is only one private  fiat-to-Bitcoin exchange \cite{yi2019new}. 
During withdraws, the client chooses one UTXO and mixes it among $k$ transactions. 
To prevent linkability by the transaction amount, it requires each withdrawal to be fixed for 1 BTC. 
And if two clients choose the same BTC, only one of them would get paid.


It follows that the natural question of anonymous cryptocurrency exchange is still open.




\medskip
\noindent{\bf Further challenges.} Besides the issues mentioned above not covered in existing studies, the anonymous cryptocurrency exchange setting has several other features that bring about more challenges: since the exchange system is always connected with external blockchain (e.g., via the {\em deposit} and {\em withdrawal} of coins), it automatically leaks highly non-trivial information (e.g., 3 BTCs has been deposited, and 2.9 BTCs has been withdrawn/transferred out 2 minutes later) such that how to best deal with them requires care.

 \noindent{\em The right anonymity/privacy goal.}\ 
From a first look, 
we may just handle the withdrawal operation and define a basic, direct anonymity notion, that breaks the link between the receiving account and user identity, and leave other operations unchanged for efficiency. 
%
A bit more formally, given two different users and a specified withdraw transaction, we can require that it is infeasible to distinguish which one conducts the withdrawal if {\em both} of them are eligible. 
However, if we examine the  {\em anonymity set} of the withdrawal, it only consists of users who have enough amount of the specified coin, which could be few. 
For example, for some unpopular assets, maybe only a very small number of users own such kinds of coins; or one user may hold a significantly larger amount of the coin than others. When a large-volume withdrawal of such token is taken place, it is easy for the platform to identify the user.



We then turn to consider stronger anonymity. One may suggest to gradually strengthening anonymity by allowing fewer unnecessary leakages (keeping some internal transaction data {\em private} such as amount) and leave seemingly safe information such as coin names as now (to avoid potentially complex solutions for protecting such info). Unfortunately, many of remaining transaction metadata, together with the inherent leakages such as 3 BTCs have been withdrawn by someone to an external address, can still reduce anonymity set. 
It is hard to have a reasonably stronger anonymity without full internal transaction privacy (excluding the inherent leakage during withdraw/deposit), as it is unclear what is the actual consequence of each specific leakage.
For these reasons, we choose a definition that insists the system does not leak anything more than necessary to the exchange platform (essentially requiring privacy). We will explain more in Sec. \ref{WithdrawAnonymitySection}.

\noindent{\it Preserving major compliance functionalities.} We also need to preserve all the critical functionalities that are currently provided by centralized exchange platforms, including compliance such as generating tax reports for users and checking sufficient reserve for the platform. \footnote{There are some related works in  accountable privacy (e.g., PGC \cite{chen2020pgc}, UTT \cite{tomescu2022utt}, Platypus \cite{Platypus}, etc),
but they only focus on the payment with a single kind of asset and enforce limits on one transaction amount or account balance or sum of all sent or received values. 
Note that these compliance requirements cannot cover the profit computation which uses the specific buying price without linking to that transaction.}

There are many types of assets/coins in an exchange system, 
and their prices fluctuate over time. 
Users gain a profit by capitalizing on the price difference between buying and selling.
It is often mandatory for users to pay taxes on their {\em accumulated} profits over time. 
At each year end, users obtain a tax report from the platform so that they can report their annual profit, e.g., to Internal Revenue Service for tax filing. 
For example, based on the suggested tax policy of Coinbase \cite{coinbasetax}, transactions that result in a tax are called taxable events. Taxable events as capital gains include selling cryptocurrency for cash, converting one cryptocurrency to another, and spending cryptocurrency on goods and services (e.g., withdrawing cryptocurrency).

In the current transparent exchange system, the platform records the whole transaction history for each account and extract easily their taxable events. 
The platform can also check the reserve easily as it knows the asset details of each account. 
This ensures that the platform possesses sufficient assets to meet the withdrawal requests of users.

However, in the anonymous setting (now also requires privacy), the platform has no idea about the asset detail of each account. 
It cannot prove solvency in the same way as before. 
Furthermore, the platform knows neither the actual profit nor the relationship between these transactions with any user. 
Without careful designs to calculate accumulated profit (without violating privacy/anonymity), some users could always claim they made no profit.

 \noindent{\em Striving for practical performance.} 
Privacy preserving constructions normally use zero-knowledge proofs. 
Although the deposit and withdrawal assets are public, the exchange details (e.g., 1 BTC for 15 ETHs) should be hidden and proven in zero-knowledge that the transaction is valid and the prices 
are recorded correctly.
In theory, zkSNARK \cite{bitansky2013succinct} may enable succinct proof size and verification time.
But the proof generation incur heavy computation for users.
$\Sigma$-protocols may also be useful, but hiding the exchanged asset types in all $n$ kinds assets usually requires communication/computation cost growing at least {\em linear} in $n$. For a practical design, we need to reduce the communication and computation overhead to be as small as possible (e.g., ideally {\em constant} cost).


\subsection{Our contributions}

\noindent\textbf{Modeling.}
We for the first time formally define the private and compliable cryptocurrency exchange.
We give a basic version of anonymity first as a warm up, which only cares about the withdrawal operation.
As we briefly discussed above, hiding only part of transaction data may not give a reasonably strong anonymity. 
In the end, we define the security model insisting that the exchange leaks essentially no information to the platform. In this way, we obtain the best possible anonymity (given that public withdrawal is always there). We also carefully define {\em soundness} properties such as {\em overdraft prevention}, and {\em compliance}. For details, we refer to Sec. \ref{model}.

\noindent\textbf{Constructions.} 
We first give a very simple construction satisfying the basic withdraw anonymity, and showcase its limitations. 
We then design the first private and compliable exchange system which is provably secure in the full private model.
Users are hidden in a large anonymity set, and they cannot withdraw more asset than they own, or report false compliance information.
To obtain full anonymity, user's information are concealed as much as possible in each transaction,
including user identities 
and the exchanged assets details. Soundness properties are ensured via efficient NIZK proofs specially designed for our purposes.
Note that proving correctness of a exchange request usually leads to a proof whose size is linear to the total number of asset types; instead, we propose an efficient construction with {\em constant} cost in both communication and computation which is independent with the number of asset types and users in the system.


\noindent{\bf{Performance evaluations.}}
We implement and evaluate our $\mathsf{\name}$  system and test the cost breakdown in each operations, and compare with those in plain exchange (without anonymity). Considering the presence of TLS communication, our overhead is minimal. We also compare with other relevant systems\footnote{\label{ucs}There is a concept of updatable anonymous credential~\cite{Bl2019UpdatableSystems}, that share similar theoretical structure of proving properties of attributes in anonymous credential; however, their main application to incentive systems supports only limited functionalities and the achieved anonymity is weak, see Sec.~\ref{IncentiveSystem}. We have to design more complicated compliance functions.} for further evidence. See Sec.~\ref{sec:eval} for details.

\section{Technical overview}
%
We first provide a high-level overview of the technique. Typically, there are two main parties involved: the platform and the user. However, in certain cases such as tax filing, there may also be an external authority involved. 

 \noindent{\em Workflow of exchange system.}
First, the user provides the real identity to the platform during registration.
Then the user interacts with the platform to deposit, exchange (e.g., 1 BTC for 15 ETHs) and withdraw asset. 
For compliance, the user generates his compliance report, gets it certified by the platform and reports to an authority. 
The platform also generates information to check its own solvency.


We use Fig \ref{overview} to visualize these (simplified) procedures.
Each interactive protocol can be expressed by \ding{192} \ding{193} \ding{194} steps.  
The user sends compliance report to the authority who verifies it in step \ding{195}.
The platform checks in step \ding{196} whether its internal state satisfies the platform compliance rule.

\begin{figure}[!htbp] 
\centering 
\includegraphics[width=0.45\linewidth, height= 0.3\linewidth]{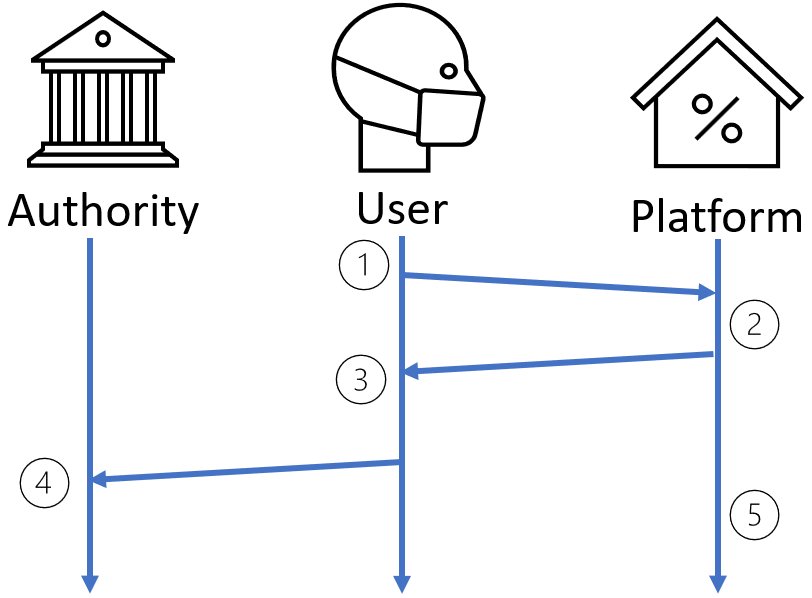} 
\caption{Overview of exchange system: \ding{192} \textit{Transaction request}; 
\ding{193} \textit{Transaction processing}: platform verifies the transaction and process it; \ding{194} \textit{Transaction completion}: user completes the transaction; \ding{195} \textit{Compliance verification}: authority verifies the compliance report of the user; \ding{196} \textit{Compliance check}: platform checks internal state with the platform compliance rule} 
\vspace{-0.2cm}
\label{overview} 
\end{figure}








 \noindent{\bf Constructions.}  Based on the workflow above, we illustrate the design idea of our efficient system $\mathsf{\name}$\/ step by step.



 \noindent {\em Basic anonymity.} As a warm-up,  
to just break the link between outgoing transaction (withdraw) and the history within the platform, we can introduce a preparation step. There, users ask for 
{\em one-time} anonymous credentials (as tokens) with amounts hidden to the platform (simply via a ``partially'' blind signature) that later can be used to do withdrawal. Users also submit a committed record containing asset details of the token for compliance purposes.  Now, the user balance becomes hidden, and users should prove that 
the sum of all hidden amounts is valid depending on his previous balance via zero-knowledge range proof. 
When withdraw, users could directly reveal such one-time credential (a valid signature on a random identifier and transaction etc), which easily prevents double-spending. Since now user balance is also hidden, all future operations including exchange, withdraw preparation will involve (efficient) zero-knowledge proofs of validity.

 \noindent{\em Full anonymity.} 
Additional protection on the exchanges and deposits is needed. Especially, the exchanges should not only be anonymous but also keep {\em asset type} and {\em amount} private. 
Each exchange transaction requires the value of exchange-in and exchange-out to be equal. To calculate the value, users need to show the used prices (now committed) are two of all current prices and correspond to the two exchanged assets. Further zero-knowledge proofs on membership and equality will be leveraged. But doing them efficiently requires care and will be explained soon in {\em practical considerations}.

\noindent{\em Supporting compliance.}
The above full anonymity construction is over-simplified, as we have not considered compliance issues. For example, since each user can have multiple credentials, he could give one credential with, say 10 Bitcoins, as a gift to any person, who may not even registered with the platform. 
Then the gift receiver could use the credential to do the anonymous trading with the platform without revealing his real-world identity, which is not compliant even with the basic KYC regulation. Moreover, we would support common compliance goals without hurting anonymity/privacy. 
In particular, we use tax filing as an example of client-compliance. 

First, all transactions from the same user should be bind together (without revealing the content) to derive accumulated profit; we thus let each user maintain {one} long-term {\em registration credential} that contains two attributes ``cost'' and ``gain'' to record the buying cost, and selling gain for exchange-out and withdraw transactions (the taxable transactions in e.g., Coinbase). Such a credential is issued when user registered to the system, and will be updated properly after each transaction. To conveniently update it, transaction metadata would also be recorded in a secure way, e.g., each coin type, buying/selling price and amount, etc (as in current exchange platform such as Coinbase). Each transaction corresponds to two one-time {\em asset credentials} w.r.t the exchanged assets which contain the corresponding trading prices and amounts.

When users request to exchange or withdraw, they need to provide proof of ownership for a valid asset credential with sufficient amounts, a valid registration credential, evidence of fair exchange, and accurate records of updated costs and gains. However, when generating the compliance report, revealing this information directly to the platform would compromise user privacy. For instance, if a user has a substantial profit, it increases the likelihood of being linked to previous large-scale withdrawals. To address it, we employ a workaround by having the platform blindly sign (thus providing validity proof) to generate the report without revealing sensitive information.

\noindent\textit{Practical considerations.} With above considerations, the users and platform have to engage in multiple non-interactive zero-knowledge proofs, some of which may be heavy if not done properly. Particularly, for each exchange transaction (e.g., user wants to buy 1 BTC using 15 ETHs), the user takes his ETH asset certificate, proves in zero-knowledge that the asset type belongs to $[n]$ via a membership proof (as asset type needs to be kept private); and proves the used prices (committed in the new asset certificates) are exactly
the current prices of the exchange-in and exchange-out assets, which also need membership proofs. To facilitate such a proof, one idea is to let the user to commit to a vector $\vec{v}$ with $n$ dimension, and prove that $\vec{v}$ is a vector of bits and contains only one entry (corresponding to his asset certificate) as $1$. Then homomorphically evaluating the linear combinations may get commitments of price$\times$amount of BTC, and ETH respectively, the user can further prove the resulting committed values are equal. The proof size is already at least $O(n)$, and computation cost even more. Recent work of one-out-of-many proof or many-of-many proof may reduce the proof size to logarithmic in $n$ but the computation cost of proof generation and verification is still (super)linear in $n$ \cite{Groth2015, diamond2021manyzether}. 




Instead, we let the platform generate signatures on each asset name and the price and make them public, called \textit{price credential}.
To capture the price fluctuation and avoid users using out-of-date price credentials, each price credential contains the timestamp of the latest price update. In the exchange transaction, the user proves that the new exchanged asset record contains the name and price and she knows the valid signature on them and the latest timestamp. 
It can be verified by the platform's {\em single} public key, and we can bring down the communication and computation cost to be constant. For details, please refer to Sec. \ref{Pricecredential}.

 \noindent\textbf{Extensions and open questions.} 
For client-compliance, there are many other regulation rules such as limiting the transaction frequency, transaction amounts, all sending or receiving amounts and scrutinizing the receiver addresses in case of financing terrorism. Our techniques can be extended very easily to also support those. See Sec. \ref{construction} for more discussions.

\noindent\textit{Solvency issues.} There are also many platform-compliance requirements. One notable one is the solvency problem that the platform should be able to check it has sufficient reserve. Now, transaction details are hidden in privacy-preserving exchange, which may increase the risk of solvency issues, and users may exchange/withdraw a large amount of certain cryptocurrency privately that exceeds the platform's reserve, thus causing a potential ``bank run''. According to the Basel Accords \cite{Basel} for the banking industry (we also use it as the platform-compliance rule in Sec. \ref{compliance-rule}), it usually requires the banks to  (i) provide sufficient liquidity (e.g., keep enough asset to cover the total withdraws of last month), and (ii) keep a sufficient minimal reserve (e.g., 0-10\% of the total assets held by all users in the platform, in the form of a major currency such as USD~\cite{IMFForeignExchangeReserve}; in our setting, Bitcoin). We show that our Pisces system with full anonymity also satisfies the first platform-compliance requirement. As the platform is still aware of the total amount of incoming/outgoing Bitcoins (and any other cryptocurrency tokens), thus the needed information could still be derived. 

For maintaining a minimal reserve, there might be some practical mitigation, e.g., actively monitoring the total withdrawal amount/pattern for each coin, limiting the exchange and withdrawal amount/frequency, etc, or involving a third-party auditor (similar to the tax authority to keep the aggregated information to manage risks). Those can be supported by extending our design. But a more rigorous solution remains open. Also, there could be even more strict and complicated rules that may require the platform to keep sufficient reserve and liquidity for every single type of coins \cite{BinanceProofofReserve}; or require the platform to generate publicly verifiable proofs of solvency. 

We remark that with our basic anonymity, the platform knows all the holdings of each account (except the link between the inside and external onchain accounts), thus can still derive all needed information for both requirements and the more strict rules.  

However, a more systematic investigation of solvency issues in the fully anonymous setting (e.g., allowing the platform to gain extra side information for solvency purposes) may again have further impact on the anonymity.

%

%
As a first step studying privacy in exchange system with efficient compliance support, there are many interesting questions and challenges to explore (e.g., supporting broader compliance rules). 
For a more systematic investigation, we leave them as interesting open problems.



\input{2preliminary}

\input{3-1syntax}
\input{3-2security}

\input{5full}

\input{6perform}

\bibliographystyle{IEEEtranS}
\bibliography{ref}


\newpage

\section*{Appendices}

\input{appendix}

\end{document}

%% file: 2preliminary.tex
\section{Preliminary}

\noindent{\emph{Notations.}} Throughout this paper, we denote with $\lambda\in \mathbb{N}$ the security parameter, and by poly($\lambda$) any function which bounded by a polynomial in $\lambda$. 
An algorithm $\mathcal{A}$ is said to be PPT if it is modeled as a probabilistic Turing machine that runs in time polynomial in $\lambda$. 
Informally, we say that a function is negligible if it vanishes faster than the inverse of any polynomial. 
A function $f : \mathbb{N} \rightarrow \mathbb{R}$ is negligible if for every positive integer $c$, there exists an integer $x_0$ such that $|f (x)| < 1/x^c$ for all $x > x_0$.
It is denoted by \textit{negl}. 
For a finite set $S$, $x \sample S$ means that $x$ is chosen uniformly from $S$.
If $n$ is an integer, $[n]$ denotes the set of positive integers $1, 2,\dots, n$.
We use $\Vec{v}$ to denote a vector.
We write $\langle \mathsf{A},
\mathsf{B}\rangle$ to denote interactive algorithms $\mathsf{A,B}$ engage in an interactive protocol, take their respective inputs, and share some transcripts.


{We briefly introduce some cryptographic primitives here for completeness and defer their details to the Appendix~\ref{AppPriliminary}.}

\noindent\textbf{Commitments.} 
A commitment scheme allows one to commit to a chosen value secretly, with the ability to only open to the same committed value later. 
A commitment scheme $\Pi_{\mathrm{cmt}}$ consists of the following PPT algorithms: 

\noindent$\mathsf{Setup}(1^{\lambda}) \rightarrow pp$: generates the public parameter \textit{pp}.\\
$\mathsf{Com}(pp, m; r) \rightarrow com$: generates the commitment for the message $m$ using the randomness \textit{r}. For ease of notation, we omit $pp$ in the input.

We require a commitment scheme to be \textit{hiding} and \textit{binding}.
A commitment is additively homomorphic if it satisfies that for any messages $m_1, m_2$ and randomnesses $r_1, r_2$: 
$\mathsf{Com}(m_1; r_1)+\mathsf{Com}(m_2; r_2)=\mathsf{Com}(m_1+m_2; r_1+r_2)$.

 \noindent\textbf{Blind signatures.}
A blind signature scheme  $\Pi_{\mathrm{bs}}$ for signing committed $n$ messages  has the following algorithms: 

\noindent$\mathsf{KeyGen}(pp) \rightarrow (pk, sk)$: takes public parameter $pp$ as input, outputs a key pair $(pk, sk)$. $pp,pk$ are implicit inputs of others.
$\mathsf{Com}(\vec{m}, r) \rightarrow c$: given messages $\vec{m}\in\mathcal{M}^n $ and randomness $r$, computes a commitment $c$. \\
$\langle\mathsf{BlindSign}, \mathsf{BlindRcv}\rangle$: it is an interactive protocol between the signer and user, with inputs $(sk, c)$ and $(\vec{m}, r)$ respectively. User outputs a signature $\sigma$. \\
$\mathsf{Vrfy}(\vec{m}, \sigma ) \rightarrow b$: it checks  $(\vec{m},\sigma)$ pair and outputs 0/1.



We require a blind signature scheme to be correct and have the properties of {\em unforgeability} and {\em blindness}.


 \noindent\textbf{Zero-knowledge argument of knowledge ($\mathsf{ZKAoK}$).}
The prover proves knowledge of $w$ such that $(x, w)$ is in some NP relation $R$. 
Here $x$ is the statement and $w$ is the witness. 
The zero-knowledge argument of knowledge \cite{lindell2003parallel} can be simulated perfectly and there exists an expected polynomial-time extractor $\mathcal{E}$ that, given black-box access to a successful prover, computes a witness $w$ with probability 1.
It is denoted by $\mathsf{ZKAoK}[(w); (x, w) \in R]$.

\section{Syntax}
In this section, we define the syntax 
that is abstracted from real exchange systems and is general for both plain and private centralized exchange systems.
Basically, an exchange system supports users depositing multiple kinds of assets (including fiat money), exchanging assets with the platform, and withdrawing assets. To comply with regulations, 
the system also checks compliance with platform rules and supports users in filing their compliance documents.





%% file: 3-1syntax.tex
\vspace{-0.2cm}
%
\subsection{Syntax}\label{syntax}
An exchange system involves three entities: the platform $\mathsf{P}$, the user $\mathsf{U}$ and an authority $\mathsf{A}$. 
The system consists of the following PPT algorithms: $\mathsf{Setup}$, $\mathsf{PKeyGen}$, 
$\mathsf{Verify}$, $\mathsf{Check}$
as well as interactive protocols:
$\langle\mathsf{Join}, \mathsf{Issue}\rangle$,  
$\langle\mathsf{Deposit}, \mathsf{Credit}\rangle$,
$\langle\mathsf{Exchange},$ $\mathsf{Update}\rangle$, $\langle\mathsf{Withdraw}, \mathsf{Deduct}\rangle$, $\langle\mathsf{File}, \mathsf{Sign}\rangle$. 
To present the syntax, we prepare some data structures.


\noindent{\textbf{{Transaction requests $\mathrmbf{reqs}$.}}}\label{remark1}
For each transaction, $\mathsf{U}$'s input includes a transaction request to specify details. We denote it as a data structure and consider five kinds of requests as follows. To keep the syntax general and simple, we add an optional attribute $aux$ to each request. $aux$ could contain several sub-attributes required by the operation but not included in the listed attributes, be different for different operations, and be specified by the detail construction.
\vspace{-0.2cm}
\setlist{nolistsep}
\begin{itemize}[noitemsep,leftmargin=*]

    \item[-] $\mathit{req_{joi}:=(info,aux)}$ denotes join request, where $\mathit{req_{joi}.info}$ is user's information for joining the system.
    \item[-] $\mathit{req_{dep}:=(name,amt,aux)}$ denotes deposit request, where $\mathit{req_{dep}.name}$ is asset name, $\mathit{req_{dep}.amt}$ is asset amount.
    \item[-] $\mathit{req_{exc}:=(name_{in},amt_{in}, name_{out}, amt_{out}),aux}$ denotes exchange request, where $\mathit{req_{exc}.name_{in}}$ is exchange-in asset name, $\mathit{req_{exc}.amt_{in}}$ is exchange-in asset amount, $\mathit{req_{exc}.name_{out}}$ is exchange-out asset name, $\mathit{req_{exc}.amt_{out}}$ is exchange-out asset amount.
    \item[-] $\mathit{req_{wit}:=(name,amt,aux)}$ is withdraw request, where $\mathit{req_{wit}.name}$ is asset name, $\mathit{req_{wit}.amt}$ is asset amount.
    \item[-] $\mathit{req_{fil}:=(uid,cp,aux)}$ denotes file request, where $\mathit{req_{fil}.uid}$ is user identifier, $\mathit{req_{fil}.cp}$ is compliance information.
\end{itemize}



\noindent{\textbf{{Transaction records $\mathrmbf{Rd_{reg}}$ and $\mathrmbf{Rd_{ast}}$.}}} \label{remark2}
Each record is an information credential pair, where the information contains several attributes. It is generated or updated during transactions, and kept privately by users. 
We denote record $\mathit{Rd}$ as a data structure and consider two kinds of records: 
the registration record $\mathit{Rd_{reg}}$ and the asset record $Rd_{ast}$.

\vspace{-0.2cm}
\setlist{nolistsep}
\begin{itemize}[noitemsep,leftmargin=*]

    \item [-] $\mathit{Rd_{reg}:=(non, uid, cp, cred)}$ denotes a registration record, including three attributes and the credential $\mathit{Rd_{reg}.cred}$ on the three attributes.  $\mathit{Rd_{reg}.non}$ is the random nonce to uniquely identify it. $\mathit{Rd_{reg}.uid}$ is the owner's unique identifier. 
    $\mathit{Rd_{reg}.cp}$ is the compliance information.  Each user holds only one valid $\mathit{Rd_{reg}}$, which is initialized in join transaction,  updated in exchange, and withdraw transaction via revoking the old one and generating a new one. 
    \item[-] $\mathit{Rd_{ast}:=(non, uid, name, amt, acp, cred)}$ denotes an asset record, including five attributes and a credential $\mathit{Rd_{ast}.cred}$ on the five attributes. Similar to $\mathit{Rd_{reg}}$, $\mathit{Rd_{ast}.non}$ is the random nonce and $\mathit{Rd_{ast}.uid}$ is the owner's identifier. $\mathit{Rd_{ast}.name}$ is the asset name, $\mathit{Rd_{ast}.amt}$ is the amount of asset, and $\mathit{Rd_{ast}.acp}$ is asset-related compliance information. Notably, in a private yet compliable setting, assets cannot be accumulated trivially in terms of quantity since they are tied to different asset kinds and compliance-related information such as selling prices. Thus each user could hold multiple asset records.
\end{itemize}

\noindent\textbf{Concrete algorithms.}
\setlist{nolistsep}
\begin{itemize}[noitemsep,leftmargin=*]
 
\item $\mathsf{Setup}$: The public parameters \textit{epp} for the exchange system is set. $\textit{epp} $ includes the public parameters for cryptographic primitives. For simplicity of syntax, we let $epp$ also include some publicly available information, such as the external blockchain, all coin prices in the exchange system, some metadata like the time, etc.

\item $\mathsf{PKeyGen}$:
$\mathsf{P}$ runs the key generation algorithm to generate a key pair $(pk, sk)$ and makes $pk$ public to users.
It initializes its internal state $st$ as $\emptyset$.




\item {$\langle\mathsf{Join}, \mathsf{Issue}\rangle$}: It is a register protocol. $\mathsf{U}$ runs the interactive algorithm $\mathsf{Join}(\textit{epp}, \textit{pk}, \mathit{req_{joi}})$, and $\mathsf{P}$ runs the interactive algorithm $\mathsf{Issue}(epp, pk, sk, st)$, where $st$ denotes the internal state of $\mathsf{P}$.  After the interaction, $\mathsf{P}$ outputs a signal bit $b$ indicating whether the operation succeeds or not, and updates its internal state to $st'$. If $b=1$, the user outputs the unique user identifier $uid$ and the registration record $\mathit{Rd_{reg}}$. 

 

\item $\langle\mathsf{Deposit}, \mathsf{Credit}\rangle$: 
It is a deposit transaction for users to deposit assets to $\mathsf{P}$. $\mathsf{P}$ runs $\mathsf{Credit}(epp, pk, sk, st)$ and 
$\mathsf{U}$ runs $\mathsf{Deposit}(epp, pk, uid,  \mathit{Rd_{reg}}, \mathit{req_{dep}})$. 
 The asset name and amount are specified in $\mathit{req_{dep}}$.
After the interaction, $\mathsf{P}$ outputs a signal bit $b$ indicating whether the operation succeeds or not, and updates its internal state to $st'$. 
If $b=1$, $\mathsf{U}$ gets a new asset record $\mathit{Rd_{ast}^{out}}$ for the deposited asset. 

\item $\langle\mathsf{Exchange}, \mathsf{Update}\rangle$: It is an exchange transaction for users to exchange assets with $\mathsf{P}$. The names and amounts of exchange-in and exchange out assets are specified by $\mathit{req_{exc}}$.
$\mathsf{U}$ runs $\mathsf{Exchange}(\textit{epp}, pk, uid, \mathit{Rd_{reg}},\mathit{Rd_{ast}}, \mathit{req_{exc}})$, 
and $\mathsf{P}$ runs $\mathsf{Update}(\textit{epp}, pk, sk, st)$.
After the interaction, $\mathsf{P}$ outputs a signal bit $b$, indicating whether the operation succeeds or not, and updates its internal state to $st'$. 
If $b=1$, $\mathsf{U}$ outputs three records: the updated ones $\mathit{Rd'_{reg}}$ and $\mathit{Rd'_{ast}}$, and a newly generated asset record $\mathit{Rd_{ast}^{out}}$ for exchange-out asset with name $\mathit{req_{exc}.name_{out}}$.

\item {$\langle\mathsf{Withdraw}, \mathsf{Deduct}\rangle$}:
It is a withdraw transaction for users to withdraw a kind of asset from $\mathsf{P}$ to the blockchain. The name and amount of withdrawn asset are specified in $\mathit{req_{wit}}$. $\mathsf{U}$ runs $\mathsf{Withdraw}(epp, pk, uid, \mathit{Rd_{reg}}, \mathit{Rd_{ast}}, \mathit{req_{wit}})$, 
and $\mathsf{P}$ runs $\mathsf{Deduct}(epp, pk, sk, st)$. 
After the interaction, $\mathsf{P}$ outputs a signal bit $b$, indicating whether the operation succeeds or not, and updates its internal state to $st'$. 
If $b=1$, $\mathsf{U}$ outputs two updated records 
$\mathit{Rd'_{reg}}, \mathit{Rd'_{ast}}$.

\item $\langle\mathsf{File}, \mathsf{Sign}\rangle$:
It is a two-party protocol in which $\mathsf{U}$ files compliance information periodically and requests $\mathsf{P}$ to sign it. $\mathsf{U}$ runs $\mathsf{File}(epp, pk,$ $uid, \mathit{Rd_{reg}}, \mathit{req_{fil}})$ and $\mathsf{P}$ runs $\mathsf{Sign}(epp, pk, sk, st)$. 
After the interaction, $\mathsf{P}$ outputs a single bit $b$, indicating whether the operation succeeds or not, and updates its internal state to $st'$. 
If $b=1$, $\mathsf{U}$ outputs an updated record $\mathit{Rd_{reg}'}$ 
and  a compliance document $\textit{doc}$ certified by $\mathsf{P}$. Similar to transaction record $\mathit{Rd}$, $\textit{doc}:= (uid,cp,mt,sig)$ is also a data structure including three attributes and a signature \textit{doc}.\textit{sig} on them, where \textit{doc}.\textit{uid} is the user identifier, \textit{doc}.\textit{cp} is the reported compliance information, and \textit{doc}.\textit{mt} is the metadata such as time.

\item$\mathsf{Verify}$: 
The authority runs $\mathsf{Verify}(\textit{epp},pk, doc)$ to check the validity of the submitted compliance document, including the consistency of metadata in $\textit{epp}$ and $\textit{doc.mt}$,  and the valid signature. It outputs a bit $b$ with $b=1$ indicating a passing check, and vice versa. 


\item$\mathsf{Check}$: $\mathsf{P}$ runs $\mathsf{Check}(epp,st)$ for self-checking the internal state's compliance with platform rules specified in $epp$. The output is a single bit $b$, with $b=1$ indicating a passing check and vice versa.

\end{itemize}

%% file: 3-2security.tex
\section{Security models\label{model}}

In this section, we formally define security models to capture the desired security properties of a private yet compliable exchange system. Along the way, we show the motivation, importance, and ideas of defining such properties. 

To the best of our knowledge, this is the first security modeling of the centralized exchange system. 
Security modeling of this work is involved in four aspects. 
(1) we put less trust in the platform than in existing plain exchange systems where the platform is always assumed to be honest. While our model gives the platform more power in some security definitions, especially for anonymity, the platform can be completely malicious;
(2) When modeling privacy/anonymity, naive attempts for a ``direct" anonymity (without privacy on other parts of transactions, or trying to strive a best balance between efficiency and hiding only part of the transactions) may not work well because of potential consequences of each seemingly benign leakage (within the exchange system).
We will elaborate on it in Sec.~\ref{WithdrawAnonymitySection};
(3) Besides desired anonymity, we also define {\em soundness} properties of overdraft prevention and client compliance security that require care too; 
(4) For platform compliance, we require that the honest platform can always self-check whether its internal state satisfies the platform compliance rule.

We first give a high-level description of the security requirements of the system.  
\noindent{\em Correctness.}
The honest user gets the correct balance amount in his account from deposit, exchange, and withdrawal, also gets the correct number of real assets from withdrawal, and gets a valid signature on his compliance information that can be verified by the authority.

\noindent{\it Anonymity.}
Given a withdraw/deposit transaction, the malicious platform should not link it to any specific user, except the user has to expose the identity, such as depositing/withdrawing fiat money from/to bank.
We start discussing it from the basic anonymity where only focusing on the withdraw or deposit transactions.
Although the basic anonymity scheme could be simple and not bring extra challenges to compliance (especially platform compliance), we show that the basic withdraw anonymity may not be sufficient, since the platform could narrow down the anonymity set based on other transactions, such as deposit, exchange, and file. Thus we further explore the best possible (full) anonymity and model it.

\noindent{\it Overdraft prevention.}
It ensures users cannot possess or spend more assets than they actually own in the system. It prevents malicious users from conducting fraudulent deposits, exchanges, or withdrawals.
	
\noindent{\it Compliance.} 
It requires that both users and the platform to comply with the regulations expressed as functions, and we call the corresponding compliance F-client-compliance and G-platform-compliance.
All entities are required to provide compliance information according to respective compliance rules, and none of them can deceive the authority with incorrect information as long as the user does not collude with the platform. 
For example, F could be a tax report function on accumulated profit, and G could be a solvency-related function on the coin reserve and liquidity. 
Tax-report-client-compliance ensures that the user cannot cheat with a value less than his latest accumulated profit this year. Solvency-platform-compliance ensures that the platform maintains 
appropriate liquidity based on the monthly assets inflow and outflow.

\subsection{Preparations for the models}
Note that the bank accounts leak the user's identity when depositing or withdrawing fiat money which is unavoidable. So we consider privacy only during cryptocurrency trading. Besides, the deanonymization attack in the network layer is out of the scope of our work. The attacker links multiple transactions by IP address, but users can protect themselves using an anonymous network like Tor \cite{tor,torbrowser}.

We provide oracles to capture the adversary's capability.
To model the capabilities of the malicious platform, 
we provide oracles: 
$\mathcal{O}^1_\mathsf{Join}$, $\mathcal{O}^1_\mathsf{Deposit}$,
$\mathcal{O}^1_\mathsf{Exchange}$,
$\mathcal{O}^1_\mathsf{Withdraw}$, 
$\mathcal{O}^1_\mathsf{File}$.
To model the capabilities of malicious users, 
we define the oracles: 
$\mathcal{O}^2_\mathsf{PKeyGen}$, $\mathcal{O}^2_\mathsf{Issue}$, $\mathcal{O}^2_\mathsf{Credit}$, $\mathcal{O}^2_\mathsf{Update}$, $\mathcal{O}^2_\mathsf{Deduct}$, $\mathcal{O}^2_\mathsf{Sign}$. 
We also provide $\mathcal{O}_{\mathsf{Public}}$ for every party to model access to some public ongoing information, such as a secure blockchain system, the prices of all assets, currencies, stocks, cryptocurrencies, and a global clock, etc.
\noindent\emph{Reference-record map $\textsc{Map}: (uid,\mathit{ref})\rightarrow Rd$:} 
When \adv\ acts as a malicious platform, it is allowed to induce honest users to conduct transactions by querying oracles. However, some oracles require specifying records as input, which are private to honest users and unavailable to \adv. To enable \adv\ to identify different records without knowing what they are,  we let \adv~ specify the reference string $\mathit{ref}$\footnote{Note that the reference string \textit{ref} used by \adv~ is different from the identifier(nonce) of the record which is privately chosen by the honest user or oracle randomly.}  for each record and Oracles keep the map $\textsc{Map}$ from key tuple $(uid,\mathit{ref})$ to value $Rd$ for \adv's later queries. For notational convenience, we let $
\textsc{Map}(uid,\mathit{ref})$ denote the record $Rd$.
In the queries, $\mathit{ref}_{\mathsf{reg}}$ is the reference for the registration record, $\mathit{ref}_{\mathsf{ast}}^{\mathsf{in}}$ is the reference for the spending asset record, and $\mathit{ref}_{\mathsf{ast}}^{\mathsf{out}}$ is the reference for the buying asset record.

\vspace{-0.2cm}
\setlist{nolistsep}
\begin{itemize}[noitemsep,leftmargin=*]
\item $\mathcal{O}_\mathsf{Public}$: when queried, it returns the public information \textit{pub}, such as the registration information, bank account, asset prices and related wallet addresses, etc. 
For all queries to other oracles, they inherently invoke $\mathcal{O}_{\mathsf{Public}}$ at first. 
We do not repeat these moves in the oracle descriptions.

\item $\mathcal{O}^1_{\mathsf{Join}}(\mathit{req_{joi}, ref}_{\mathsf{reg}})$: it interacts with $\mathcal{A}$\ by running the protocol $\langle \mathsf{Join}, \mathsf{Issue} \rangle$, where oracle runs $\mathsf{Join}(epp, pk, \mathit{req_{joi}}) \rightarrow (uid, \mathit{Rd_{reg}})$. 
If $\mathsf{Join}$ algorithm outputs $\bot$, then oracle outputs $\bot$.
Otherwise, oracle  adds $(\mathit{uid, ref}_{\mathsf{reg}}, \mathit{Rd_{reg}})$ to $\textsc{Map}$ and outputs $uid$ to \adv. 


    \item $\mathcal{O}^1_\mathsf{Deposit}(uid, \mathit{req_{dep}}, \mathit{ref}_{\mathsf{reg}}, \mathit{ref}_{\mathsf{ast}}^{\mathsf{out}})$: oracle first gets record $\mathit{Rd_{reg}=}\textsc{Map}(uid,\mathit{ref}_{\mathsf{reg}})$ from \textsc{Map} per references. Then it interacts with $\mathcal{A}$ by running $\langle \mathsf{Deposit}, \mathsf{Credit} \rangle$ protocol, where oracle runs $\mathsf{Deposit}(epp, pk, uid, \mathit{Rd_{reg}},\mathit{req_{dep}})\rightarrow (\mathit{Rd'_{reg}},\mathit{Rd_{ast}^{out}})$. If $\mathsf{Deposit}$ algorithm outputs $\bot$, then oracle outputs $\bot$; otherwise, oracle updates the map by setting  $\textsc{Map}(uid,\mathit{ref}_{\mathsf{reg}}) \leftarrow \mathit{Rd'_{reg}}$ and adds a new tuple $(uid,\mathit{ref}_{\mathsf{ast}}^{\mathsf{out}} ,\mathit{Rd_{ast}^{out}})$ to \textsc{Map}.  \adv\ gets interaction transcripts but no more output from oracle.

  \item $\mathcal{O}^1_\mathsf{Exchange}(uid,\mathit{req_{exc}},\mathit{ref}_{\mathsf{reg}}, \mathit{ref}_{\mathsf{ast}}^{\mathsf{in}}, \mathit{ref}_{\mathsf{ast}}^{\mathsf{out}})$: oracle first gets records $\mathit{Rd_{reg}}=\textsc{Map}(uid,\mathit{ref}_{\mathsf{reg}})$, $ \mathit{Rd_{ast}}=\textsc{Map}(uid,$ $\mathit{ref}_{\mathsf{ast}}^{\mathsf{in}})$ from \textsc{Map} per references. Then it interacts with $\mathcal{A}$ by running $\langle \mathsf{Exchange}, \mathsf{Update} \rangle$ protocol, where oracle runs $\mathsf{Exchange}(epp, pk, uid, \mathit{Rd_{reg}},\mathit{Rd_{ast}},\mathit{req_{exc}})\rightarrow (\mathit{Rd'_{reg}},\mathit{Rd'_{ast}},\mathit{Rd_{ast}^{out}})$. If $\mathsf{Exchange}$ algorithm outputs $\bot$, then oracle outputs $\bot$; otherwise, oracle updates the map by setting  $\textsc{Map}(uid,\mathit{ref}_{\mathsf{reg}}) \leftarrow \mathit{Rd'_{reg}}$ and $\textsc{Map}(uid,\mathit{ref}_{\mathsf{ast}}^{\mathsf{in}}) \leftarrow \mathit{Rd'_{ast}}$, and adds a new tuple $(uid,\mathit{ref}_{\mathsf{ast}}^{\mathsf{out}} ,\mathit{Rd_{ast}^{out}})$ to \textsc{Map}. \adv\ gets interaction transcripts but no more output from oracle.

    \item $\mathcal{O}^1_\mathsf{Withdraw}(uid,\mathit{req_{wit}},\mathit{ref}_{\mathsf{reg}}, \mathit{ref}_{\mathsf{ast}}^{\mathsf{in}})$: 
    oracle first gets records $\mathit{Rd_{reg}}=\textsc{Map}(uid,\mathit{ref}_{\mathsf{reg}})$, $ \mathit{Rd_{ast}}=\textsc{Map}(uid,\mathit{ref}_{\mathsf{ast}}^{\mathsf{in}})$ from \textsc{Map} per references. Then it interacts with $\mathcal{A}$ by running $\langle \mathsf{Withdraw}, \mathsf{Deduct} \rangle$ protocol, where oracle runs $\mathsf{Withdraw}(epp, pk, uid, \mathit{Rd_{reg}},\mathit{Rd_{ast}},\mathit{req_{wit}})\rightarrow (\mathit{Rd'_{reg}},\mathit{Rd'_{ast}})$. If $\mathsf{Exchange}$ algorithm outputs $\bot$, then oracle outputs $\bot$; otherwise, oracle updates the map by setting  $\textsc{Map}(uid,\mathit{ref}_{\mathsf{reg}}) \leftarrow \mathit{Rd'_{reg}}$, $\textsc{Map}(uid,\mathit{ref}_{\mathsf{ast}}^{\mathsf{in}}) \leftarrow \mathit{Rd'_{ast}}$.  \adv\ gets interaction transcripts but no more output from oracle.
    

    \item $\mathcal{O}^1_\mathsf{File}(uid, \mathit{req_{fil}}, \mathit{ref}_{\mathsf{reg}})$: oracle first get records $\mathit{Rd_{reg}=}\textsc{Map}(uid,\mathit{ref}_{\mathsf{reg}})$ from \textsc{Map} per references. Then it interacts with $\mathcal{A}$ by running $\langle \mathsf{File}, \mathsf{Sign} \rangle$ protocol, where oracle runs $\mathsf{File}(epp, pk, uid, \mathit{Rd_{reg}},\mathit{req_{fil}})\rightarrow (\mathit{Rd'_{reg}, doc})$. If $\mathsf{File}$ algorithm outputs $\bot$, then oracle outputs $\bot$; otherwise, oracle updates the map by setting  $\textsc{Map}(uid,\mathit{ref}_{\mathsf{reg}}) \leftarrow \mathit{Rd'_{reg}}$. \adv\ gets interaction transcripts but no more outputs from oracle.

    \item $\mathcal{O}^2_\mathsf{PKeyGen}$: It can only be invoked once. When triggered, run $(pk, sk) \leftarrow \mathsf{PKeyGen}(epp)$. 
    It initializes the internal state as $ st \leftarrow \emptyset$.
	It outputs $pk$.
	
	\item $\mathcal{O}^2_\mathsf{Issue}$: 
	$\mathcal{A}$ runs $\mathsf{Join}$ algorithm and interacts with the $\mathcal{O}^2_\mathsf{Issue}$ oracle. 
 $\mathcal{O}^2_\mathsf{Issue}$ runs $\mathsf{Issue}$ algorithm, takes $(\textit{epp}, pk, sk)$ as input, and receives user's transcript $\mathit{ts}$ as external input. 
It outputs a signal bit $b$ indicating whether the operation succeeds or not.
If $b=0$, it outputs $\bot$. 
 
	
	\item $\mathcal{O}^2_\mathsf{Update}$: 
	$\mathcal{A}$ runs $\mathsf{Exchange}$ algorithm and interacts with the $\mathcal{O}^2_\mathsf{Update}$ oracle. $\mathcal{O}^2_\mathsf{Update}$ runs $\mathsf{Update}$ algorithm, takes $(epp, pk, sk)$ as input, and receives user's transcript $\mathit{ts}$ as external input. It outputs a signal bit $b$ indicating whether the operation succeeds or not. If $b=0$, it outputs $\bot$. 
	
    \item $\mathcal{O}^2_\mathsf{Credit}$:
    it is similar to $\mathcal{O}^2_\mathsf{Update}$ except that here they run the $\langle \mathsf{Deposit}, \mathsf{Credit} \rangle$ protocol and \adv\ gets $\{\mathit{Rd_{ast_i}}\}$.
    
    \item $\mathcal{O}^2_\mathsf{Deduct}$: 
    it is similar to $\mathcal{O}^2_\mathsf{Update}$ except that here they run $\langle \mathsf{Withdraw}, \mathsf{Deduct} \rangle$ and \adv\ gets $\mathit{\{Rd'_{reg},Rd'_{ast_i}\}}$.
	
    \item $\mathcal{O}^2_\mathsf{Sign}$: it is similar to $\mathcal{O}^2_\mathsf{Update}$ except that here they run the $\langle \mathsf{File}, \mathsf{Sign} \rangle$ protocol and \adv\ gets $\mathit{doc}$.
\end{itemize}

\input{4warmup}

\subsection{Soundness definitions}

\noindent\textit{Extractor.} 
In overdraft prevention and client-compliance experiments, adversary \adv\, who acts as a malicious user, gets some valid records after querying oracles and keeps them secret. 
It means that after some successful anonymous transactions, the experiment does not know how many assets the users actually own and their correct compliance information.
So \textit{it is hard to decide whether \adv~breaks the overdraft prevention or compliance properties}.
To deal with this dilemma, in those security experiments, we introduce an extractor \edv\ that can output the user identity and detailed information for each transaction. 
Note that both overdraft prevention and compliance are soundness properties. 
We mimic the classic proof-of-knowledge style of definition, and the extractor can rewind \adv\ to the former state and \adv\ reuses its randomness $r_{\adv}$, similar to the proof of knowledge extractor \cite{lindell2003parallel}. 
Then the experiment is able to check if any overdraft or compliance cheating happens.

\subsubsection{Overdraft prevention}
 
Overdraft prevention requires that users cannot spend more than they own within the platform. 
Concretely it ensures no malicious users could exchange or withdraw more assets than they actually own. 
Using the transaction details extracted by the extractor \edv, the experiment can check whether an overdraft happens:
(1) the user gets credited more assets than his deposit or exchange-in; 
(2) the user gets deducted less asset than his withdrawal or exchange-out; 
(3) the remainder amount of asset is negative;
(4) the exchange is unfair; 
(5) the user steals others' asset.

We formally define overdraft prevention via the following experiment. 
$\adv$ acts as malicious users and interacts with extractor $\edv$ via querying oracles: 
$\mathcal{O}_{\mathsf{od}}=\{\mathcal{O}^2_{\mathsf{Issue}},\mathcal{O}^2_{\mathsf{Credit}},\mathcal{O}^2_{\mathsf{Deduct}},\mathcal{O}^2_{\mathsf{Update}},$ $\mathcal{O}^2_{\mathsf{Sign}},\mathcal{O}_{\mathsf{public}}\}$.
\adv\ can query at most $N=poly(\lambda)$ times, then it halts. 
$\mathcal{E}$ extracts a set of successful transaction histories $\{h_t\}$ for $t\in[N]$, where each transaction history $h_t=(uid,$$\mathit{Rd_{reg}},$$ \mathit{Rd_{ast}},$$\mathit{Rd'_{reg}},$ $\mathit{Rd'_{ast}}, \mathit{Rd_{ast}^{out}},\mathit{ts_t},\mathit{pub_t})$ includes user id $uid$, the input records $(\mathit{Rd_{reg}}, \mathit{Rd_{ast}})$, the output records $(\mathit{Rd'_{reg}}, \mathit{Rd'_{ast}},$ $\mathit{Rd_{ast}^{out}})$, the transaction transcript $\mathit{ts_t}$, and the related public information $\mathit{pub_t}$, where some records could be empty for some transactions. For example, $\mathit{Rd_{ast}^{out}}$ is empty in withdraw transaction.
Especially, transaction transcript $\textit{ts}_t:=(\textit{name}, \textit{amt},\dots )$ is a tuple of attributes including the asset name $\textit{ts}_t.\textit{name}$ and amount $\textit{ts}_t.\textit{amt}$, etc. Public information $\mathit{pub_t}:=(pr_{in}, pr_{out},\dots)$ is a tuple of attributes including input-asset price $\mathit{pub_t}.pr_{in}$, output-asset price $\mathit{pub_t}.pr_{out}$, etc. Please note, there could be some other metadata per the implementation need, so we cannot specify all the attributes and some attributes could be empty for different transactions. 
Finally, the experiment sequentially checks each transaction history to figure out whether any one of the above overdraft cases happens. 
Especially, in deposit and withdraw transactions, $\mathit{ts_t}$ contains the deposited or withdrawn asset information: $\textit{ts}_t.\textit{name}=i$, $\textit{ts}_t.\textit{amt}=k_i$ denotes that the user deposits or withdraws the asset $i$ with amount  $k_i$. 
To facilitate the check, the experiment maintains a list $\mathsf{RdSet}$ for tracking asset records that have not been spent till the checkpoint, which is initialized as empty.

\begin{figure}[htbp]
    \centering
    \fbox{\procedure{$\mathrm{Exp}^{\mathrm{od}}(\mathcal{A},\mathcal{E},\lambda)$}{		\textit{epp}\leftarrow\mathsf{Setup}(\mathcal{G}(1^{\lambda})), (1^n,st)\leftarrow\mathcal{A}(\textit{epp}), \mathrm{for\ some}\ n\in\mathbb{N}\\
    (pk,sk)\leftarrow\mathsf{PKeyGen}(epp,1^n)\\
    \mathrm{Run}\ \mathcal{A}^{\mathcal{O}_{\mathsf{od}}}(\textit{epp},pk,st) \\
    \ \pcif \text{any oracle aborts} \pcthen \pcreturn 0;\\ 
    \ \pcelse \text{continue until \adv\ halts}\\
    \mathrm{Run}\ \{\mathit{h_t}\}\leftarrow\mathcal{E}^{\mathcal{A}}(\textit{epp})\\
    \ \hfill{\color{brown}{\text{// \edv\ could control the randomness of \adv }}}\\
    \mathrm{Set}\  \mathsf{RdSet}\leftarrow \emptyset \\
    \mathrm{For}\ t=1\ \text{to } N, \text{check}\ \mathit{h_t}:\\
    \ \text{Parse}\ \mathit{h_t=(uid,}\mathit{Rd_{reg}}, \mathit{Rd_{ast}},\mathit{Rd'_{reg}}, \mathit{Rd'_{ast}}, \mathit{Rd_{ast}^{out}},\mathit{ts_t},\mathit{pub_t})\\
  \ \text{For}\ \mathsf{Deposit}\ \text{transaction}:\\ 
  \quad \pcif  Rd_{\textit{ast}}^{\mathit{out}}.\textit{name}\neq\mathit{ts_t}.\textit{name}\ \text{or}\ Rd_{\textit{ast}}^{\mathit{out}}.\textit{amt}\neq\mathit{ts_t}.\textit{amt}\\ 
  \quad \pcthen \pcreturn 1\\
   \quad \color{brown}{\text{// the name or amount of  credited asset record is wrong}}\\
  \quad \pcelse \text{let}\ \mathsf{RdSet}\leftarrow\{\mathit{Rd'_{reg}}, \mathit{Rd_{ast}^{out}}\}\cup\mathsf{RdSet}\\
  \ \text{For}\ \mathsf{Exchange}\ \text{transaction}:\\ 
  \quad \pcif \text{any of the followings happens}, \pcthen \pcreturn 1:\\
  \quad - \{\mathit{Rd_{reg},Rd_{ast}}\}\not\subseteq \mathsf{RdSet}; \color{brown}{\text{// invalid records}}\\
  \quad - Rd'_{\textit{ast}}.\textit{amt}<0\ \color{brown}{\text{// deducted amount exceeds asset amount} }\\
  \quad - \mathit{(Rd_{ast}.amt-Rd'_{ast}.amt)\cdot pub_t.pr_{in} \neq}\\
  \qquad\qquad \quad \mathit{Rd_{ast}.amt \cdot pub_t.pr_{out}}\\
  \qquad \color{brown}{\text{// the deducted value is not equal to the credited value}}\\
  \quad \pcelse  \mathsf{RdSet}\leftarrow \mathsf{RdSet}\setminus\{ 
  \mathit{Rd_{reg},Rd_{ast}}
  \}\cup \{\mathit{Rd'_{reg},Rd'_{ast},Rd_{ast}^{out}}\}\\
  \ \text{For}\ \mathsf{Withdraw}\ \text{transaction}:\\ 
  \quad \pcif \text{any of the followings happens}, \pcthen \pcreturn 1:\\
  \quad - \{\mathit{Rd_{reg},Rd_{ast}}\}\not\subseteq \mathsf{RdSet};\\
  \quad - \mathit{Rd_{ast}.name\neq ts_t.name}\\ 
  \quad - \mathit{Rd'_{ast}.amt<0}\ \text{or}\ \mathit{Rd_{ast}.amt-Rd'_{ast}.amt< ts_t.amt}\\
  \qquad\color{brown}{\text{// withdraws more asset than the deducted amount}};\\
  \quad \pcelse \text{let}\ \mathsf{RdSet}\leftarrow\mathsf{RdSet}\setminus\{\mathit{Rd_{reg},Rd_{ast}}\}\cup \{\mathit{Rd'_{reg},Rd'_{ast}}\}\\
  \pcreturn 0 
		}		
	}
	\caption{Overdraft prevention experiment.}
        \vspace{-0.2cm}
	\label{odp}
\end{figure}

\begin{definition}[Overdraft Prevention]\label{od} 
As shown in Figure \ref{odp}, we say that an exchange system can prevent overdraft if for all PPT $\mathcal{A}$ and $\lambda$, there exists \edv\ such that it holds that \begin{equation*}
	\mathrm{Pr}[\mathrm{Exp^{od}}(\mathcal{A},\edv, \lambda)=1]\leq negl(\lambda)
\end{equation*}

\end{definition}

\vspace{-0.2cm}	

\subsubsection{Compliance}
This property requires both clients and the platform to comply with the regulation rules. Here we represent these rules using compliance functions, and formalize both client compliance and platform compliance.

In the client compliance experiment, 
$\adv$ acts as malicious users and interacts with extractor $\edv$ via querying oracles: 
$\mathcal{O}_{\mathsf{clie-comp}}=\{\mathcal{O}^2_{\mathsf{Issue}},\mathcal{O}^2_{\mathsf{Credit}},\mathcal{O}^2_{\mathsf{Deduct}},\mathcal{O}^2_{\mathsf{Update}}, \mathcal{O}^2_{\mathsf{Sign}}, \mathcal{O}_{\mathsf{public}}\}$. 
\adv~outputs a certified document $\textit{doc}^*$ with four attributes $ (uid,cp,mt,sig)$.
\edv~ extracts a set of successful transaction histories $\{h_t\}$ as in overdraft prevention experiment. \adv\ wins, if $\textit{doc}^*$ passes the authority verification, i.e., $\mathsf{Verify}(epp,pk,\textit{doc}^*)\rightarrow 1$,  but there exists extracted transaction history that does not follow basic client-compliance rules (we will specify in the following), or the submitted valid document $doc^*$ is inconsistent with the extracted transaction histories $\{h_t\}$.
Concretely, each transaction in $\{h_t\}$ satisfy that: 
(1) the user has already registered (KYC rule); 
(2) all records in one transaction belong to the same user (AML rule to avoid secretly transferring assets to other accounts);
(3) the compliance-related information in each asset record, such as buying and selling price, is correct (general compliance rule).
To facilitate the check, the experiment maintains a list $\mathsf{RU}$ of all registered users, which is initialized as empty.

If all transaction histories in $\{h_t\}$ pass the above check, consistency checks between $\textit{doc}^*$ and $\{h_t\}$ per function $F$ will also be done. 
Specifically, in one transaction, the compliance information $\mathsf{cp}_{t}$ is collected from $\{h_t\}$ (e.g., the asset prices and amount) and is added to the user's the compliance information set $\{\mathsf{cp}\}_{uid}$. 
Then $F$ is applied to $\{\mathsf{cp}\}_\mathit{doc^*.uid}$  to get the final result $\widetilde{cp}_\mathit{doc^*.uid}$.
See Fig~\ref{cp1} for details.

\begin{figure}[htbp]
	\centering
	\fbox{
		\procedure{$\mathrm{Exp}^{\mathrm{clie-comp}}(\mathcal{A}, \mathcal{E}, F, \lambda)$}{
		epp\leftarrow\mathsf{Setup}(\mathcal{G}(1^{\lambda})), (1^n,st)\leftarrow\mathcal{A}(epp), \mathrm{for\ some}\ n\in\mathbb{N}\\
		(pk,sk)\leftarrow\mathsf{PKeyGen}(epp,1^n), \mathsf{RU} \leftarrow \emptyset\\
  \mathrm{Run}\ \mathit{doc^*}:=(uid,cp,mt,sig)/\bot\leftarrow  \mathcal{A}^{\mathcal{O}_{\mathsf{clie-comp}}}(epp,pk,st) \\
		\pcif \mathrm{oracle\ aborts\ or\ }  \mathsf{Verify}(epp, pk, \mathit{doc}^*)=0 
  \pcthen \pcreturn 0\\ 
        \mathrm{Run}\  \{h_t\}\leftarrow\mathcal{E}^{\mathcal{A}}(epp) \
        \color{brown}{\text{// \edv\ controls the randomness of \adv}} \\
        \mathrm{For}\ t=1\ \text{to } N, \text{check}\ \mathit{h_t}:\\
        \ \mathrm{Parse}\ \mathit{h_t=(uid,}\mathit{Rd_{reg}}, \mathit{Rd_{ast}},\mathit{Rd'_{reg}}, \mathit{Rd'_{ast}}, \mathit{Rd_{ast}^{out}},\mathit{ts_t},\mathit{pub_t})\\
        \ \text{For}\ \mathsf{Join}\ \text{transaction}:\\ 
        \quad 
        \text{let}\ \mathsf{RU}\leftarrow\{uid\}\cup\mathsf{RU},\{\mathsf{cp}\}_{uid}=\emptyset\\ 
        \ \text{For}\ \mathsf{Deposit}\ \text{transaction}:\\ 
        \quad \pcif \text{any of the followings happens}, \pcthen \pcreturn 1:\\
  \quad - \text{single transaction involves different user identifiers};\\
  \quad - \textit{uid}\notin\mathsf{RU};\color{brown}\\
  \quad\quad \color{brown}{\text{// also check them in exchange and withdraw transactions}}\\
  \quad - \mathit{Rd_{ast}^{out}.acp\neq pub_t.pr_{out}}
  \qquad \color{brown}{\text{// price was wrong}}\\
  \ \text{For}\ \mathsf{Exchange}\ \text{transaction}:\\ 
  \quad \pcif \mathit{Rd'_{ast}.acp\neq Rd_{ast}.acp \ \text{or} \ Rd_{ast}^{out}.acp\neq pub_t.pr_{out}}\\
  \qquad \pcthen \pcreturn 1\\
  \quad \pcelse \text{collect}\ \mathsf{cp}_t \text{ from } h_t, \text{add } \mathsf{cp}_t \text{ to } \{\mathsf{cp}\}_{uid}\\
  \ \text{For}\ \mathsf{Withdraw}\ \text{transaction}:\\ 
  \quad \pcif Rd'_{\textit{ast}}.acp\neq Rd_{\textit{ast}}.acp \pcthen \pcreturn 1\\
  \quad \pcelse \text{collect}\ \mathsf{cp}_t \text{ from } h_t, \text{add } \mathsf{cp}_t \text{ to } \{\mathsf{cp}\}_{uid}\\
  \pcif \{\mathsf{cp}\}_\mathit{doc^*.uid}=\emptyset, \pcthen \pcreturn 0\\
  \pcelse \text{compute } \widetilde{cp}_\mathit{doc^*.uid}\leftarrow F(\{\mathsf{cp}\}_\mathit{doc^*.uid})\\
  \color{brown}{\text{// \textit{F} is a function specified by the compliance rule}}\\
    \pcif  doc^*.cp\neq \widetilde{cp}_\mathit{doc^*.uid} \pcthen \pcreturn 1 \\
    \pcelse \pcreturn 0
		}
	}
	\caption{F-Client-Compliance experiment.}
\vspace{-0.4cm}
	\label{cp1}
\end{figure}
\begin{definition}[Client Compliance]\label{defcp} 
The client compliance experiment is shown in Figure \ref{cp1}. We say that an exchange system is client-compliant w.r.t. a compliance function F if for all PPT $\mathcal{A}$ and $\lambda$, there exists \edv\ such that
\begin{equation*}
	\mathrm{Pr}[\mathrm{Exp^{clie-comp}}(\mathcal{A},\edv, F, \lambda)=1]\leq negl(\lambda)
\end{equation*}

\end{definition}

For {\em platform compliance}, it is similar with the correctness, and we require that the internal state of the honest platform is always satisfied with the platform compliance rule. 
For example, the platform can self-check whether it owns sufficient cash and assets to cover fund outflows for the previous 30 days according to the regulation rule \cite{Basel}. 

%% file: 4warmup.tex
\label{WithdrawAnonymitySection}

\subsection{Basic anonymity}
Basic anonymity guarantees that even a malicious platform cannot link the wallet address with any honest user. 
It consists of basic withdraw anonymity and deposit anonymity.

\noindent\textit{Basic withdraw anonymity.}
We define the model in Figure~\ref{anoy}, the adversary $\mathcal{A}$ interacts with any honest user by querying the anonymity oracle set:  $\mathcal{O}_\mathsf{anony}=\{\mathcal{O}^1_\mathsf{Join}$, $\mathcal{O}^1_\mathsf{Deposit},\mathcal{O}^1_\mathsf{Exchange},\mathcal{O}^1_\mathsf{Withdraw},\mathcal{O}^1_\mathsf{File}$, $\mathcal{O}_\mathsf{Public}\}$ oracles.
The adversary submits $(uid_0, uid_1, \mathit{ref}_{\mathsf{ast}}^0,$ $\mathit{ref}_{\mathsf{ast}}^1, \mathit{req_{wit}})$ as the challenge.
It also outputs some internal state information \textit{st}.  


\begin{figure}[htbp]
	\centering
	\fbox{
		\procedure{$\mathrm{Exp}^{\mathrm{ano-wit}}(\mathcal{A},\lambda)$}
		{\textit{epp}\leftarrow\mathsf{Setup}(\mathcal{G}(1^{\lambda}))\\
		(pk,st)\leftarrow\mathcal{A}(\textit{epp})\\	(uid_0, uid_1, \mathit{ref}_{\mathsf{ast}}^0, \mathit{ref}_{\mathsf{ast}}^1, \mathit{req_{wit}}, st)\leftarrow\mathcal{A}^{\mathcal{O}_\mathsf{anony}}(st)\\
        \pcif 
        \textsc{Map}(uid_0, \mathit{ref}_{\mathsf{ast}}^0)=\bot \ or \ \textsc{Map}(uid_1, \mathit{ref}_{\mathsf{ast}}^1) =\bot \\ \quad \pcreturn 0
        \quad \color{brown}{\text{//no record mapped by references $\mathit{ref}_{\mathsf{ast}}^0, \mathit{ref}_{\mathsf{ast}}^1$}}\\
        \pcif \textit{req}_{\textit{wit}}.\textit{amt}\neq\textsc{Map}(uid_0, \mathit{ref}_{\mathsf{ast}}^0).\textit{amt}\ \mathrm{or}\\ \quad \textit{req}_{\textit{wit}}.\textit{amt}\neq\textsc{Map}(uid_1, \mathit{ref}_{\mathsf{ast}}^1).\textit{amt}\ \pcreturn 0\\
		\pcelse\ b\sample\{0,1\}\\
            \mathrm{Interacts\ with}\ \mathcal{A} \mathrm{\ by\ running} \\		
            \quad \mathsf{Withdraw}(\textit{epp},\textit{pk},\textit{uid}_{b},\textsc{Map}(uid_b, \mathit{ref}_{\mathsf{ast}}^b), \mathit{req_{wit}})\\
		\hat{b}\leftarrow \mathcal{A}^{
        \mathcal{O}^*_\mathsf{anony}}(st)\\ 
        \quad \color{brown}{\text{// * requires no query with references $\mathit{ref}_{\mathsf{ast}}^0,\mathit{ref}_{\mathsf{ast}}^1$ }}\\
		\pcreturn (\hat{b}==b)
		}
		}
	\caption{Basic withdraw anonymity experiment}
 \vspace{-0.2cm}
\label{anoy}
\end{figure}

\begin{definition}[Basic withdraw anonymity]\label{defano} 
We say that an exchange system provides basic withdraw anonymity if for all PPT $\mathcal{A}$ and $\lambda$, in the experiment shown in Fig. \ref{anoy}, it holds that
\begin{equation*}
	|\mathrm{Pr}[\mathrm{Exp^{ano-wit}}(\mathcal{A},\lambda)=1]-{1}/{2}|\leq negl(\lambda)
\end{equation*}

\end{definition}

\vspace{-0.2cm}	
\noindent{\it Warm-up construction.}
To achieve the basic withdraw anonymity, an intuitive idea is 
cutting the link to the user's real identity within the withdraw operation, and leaving all other operations plain. Our warm-up construction follows this simple idea by partitioning the withdraw operation into two separate steps: first, users log in their plain account and request for a one-time anonymous credential (just use blind signatures) on the coin they plan to withdraw; 
second, they could show an anonymous credential without login to get the asset withdrawn on-chain. 
If the withdrawn amount is arbitrary and different users withdraw different amounts of assets, the special amount helps the platform identify a specific withdrawal.
To handle this problem, 
our method is hiding the withdrawn amount in the first step where the user request the anonymous credential with a committed amount.  
We introduce a brief idea here and defer the detailed description to the Appendix~\ref{simple}.

For example, Alice has 50 BTCs in her account and Bob has 100 BTCs.
Alice wants to withdraw 5 BTCs and Bob wants to withdraw 2 BTCs.
Firstly, they commit on these values and prove they have enough balance and request the platform to issue credentials. 
Once the proof gets verified, the platform signs blindly and stores these commitments in their accounts.
Alice unblinds the signature and shows it to the platform for withdrawing 5 BTCs.
The platform cannot distinguish it is Alice or Bob since both of them have enough BTCs and have requested.
Afterwards, if they want to withdraw or exchange, they should prove that they have enough balance after deducting all committed amounts from the plain balance.

\noindent\textit{Basic deposit anonymity.} 
This property prevents the platform from linking the user's past on-chain transactions to his identity via deposit operations. 
Modeling it can be regarded as a symmetric work of basic withdraw anonymity except that all deposit requests are achievable naturally for any user.

To add deposit anonymity, one more one-time anonymous credential can be employed.
Its workflow is an inverse version of anonymous withdraw. 
When depositing assets, the user, as an anonymous guest, initializes the process by requesting the platform to issue a one-time use anonymous credential, which contains the asset name and amount.
Then he requires the platform to credit the asset balance to his account using that credential without showing the asset details.

\noindent{\it Limitations of basic withdraw anonymity.} 
The above constructions are efficient and satisfy the basic anonymity model, but we observe that the anonymity is limited.

It is well-known that the anonymity strength depends on the size of the anonymity set. The greater the anonymity set is, the higher the level of anonymity a user can achieve. When considering the anonymity set for a withdrawal transaction, it comprises users who can withdraw from the view of the platform.
In the above scheme, anonymous credentials are requested from real-name accounts. 
The anonymity set consists of users who have requested credentials for the same asset and their account balance exceeds the withdrawn amount. The following example narrows down the set size to one. 

\noindent{\em Example 1:} 
{Alice deposits 50 BTCs, Bob deposits 100 XRPs and 100 BTCs, and Clare deposits 1000 BTCs. Only Alice and Bob request anonymous credentials for BTCs. Later, a withdrawal of 51 BTCs occurs, it can thus be linked to Bob as he is in possession of a sufficient amount of BTCs.}


\subsection{Full anonymity}
\label{full}
As we mentioned above, the anonymity set of basic anonymity could be quite small. 
So we explore the stronger anonymity of the withdrawal. 
We first attempt to get perfect anonymity with all users in the anonymity set. 
Unfortunately, it is impossible due to some \textit{unavoidable leakage}.
We will elaborate it later.
For other kinds of leakage, we check that whether it is even worth to prevent, as any protective measure comes with cost. 
After some attempts, it turns out that any other leakage could be used to reduce the anonymity set.
We explain that with some examples.
Finally, we define the best possible anonymity called interactive indistinguishability by by constraining the information leakage to the minimum. 

 \noindent{\it {Stronger anonymity is needed.}}
Basic withdraw anonymity only ensures the anonymity set includes those eligible users who own enough of the withdrawn asset and are capable to withdraw.  
It is acceptable for some popular assets that a lot of people own and the withdrawn amount is small such that the anonymity set is large enough.
But it rules out many interesting scenarios, such as withdrawing some special assets owned by a small number of people or a comparatively large amount of assets that few people have so much.
Thus we aim to explore a stronger model which provides larger anonymity set.


\noindent{\em {Perfect anonymity is impossible.}} In the ideal case, the anonymity set of each withdraw transaction consists of all registered users in the system which is called the perfect anonymity. 
Unfortunately, it cannot be true since the platform always can exclude some users using some public information. 
For example, given a withdrawal of 100 Bitcoins and a newly enrolled user Alice, she is in no way the user of this withdrawal if there is no such big amount deposit in the system after her registration.


Due to the special setting of exchange platform, some information is unavoidably public to the platform, which we call \textit{unavoidable leakage}, like transaction types (deposit or exchange or withdraw), users' registration information (due to KYC requirement), deposited and withdrawn asset details (the asset name and amount, bank accounts, and wallet addresses), and even some out-of-band information like the users' behavioral preference.

To achieve the best possible anonymity, it seems that only the unavoidable leakage is acceptable. 
But a series of natural questions are: why do we need to hide so many? Can we leak a little bit more, like the privacy (identity, coin name, and amount) of the exchange? Does it hurt the best possible anonymity?


\noindent{\it {Necessity of privacy and towards best possible anonymity.}} 
To answer the above questions, we identify the avoidable leakage information which can be concealed using some cryptographic tools. 
Concretely, we assort the avoidable leakage into five classes according to the transaction: the identity in depositing coins, the identity in exchange, the contents in exchange including the coin name and amount, and the identity in withdrawing coins, and the compliance information in filing operation. 
We hide each of them and leak the other part to test whether the anonymity set is affected.

It is easy to see the identity in withdrawal cannot be leaked.  
For the three leakages occurring in the deposit and exchange, we show an example of the exchange system with a series of transactions and check the anonymity set if any avoidable leakage is allowed.

\noindent{\em Example 2:} 
In a cryptocurrency exchange system, 
there are a bunch of users registered and doing transactions, then David and Ella joined. After that somebody (David) deposits 10 BTCs. Then there is an exchange transaction: somebody (Alice, a registered user) exchanges some coins (2 BTCs to 20 ETHs). 
Then a withdrawal happens: somebody withdraws 5 XRPs. 
Check its anonymity set: 
\vspace{-0.2cm}
\setlist{nolistsep}
\begin{itemize}[noitemsep,leftmargin=*] 
    \item[1.] If the identity of deposit is leaked, then the platform knows that David deposits 10 BTCs, and Ella is excluded from the anonymity set and David is included.
    \item[2.] If the identity of exchange is leaked, the platform knows that David and Ella cannot withdraw 5 XRPs, then both David and Ella are excluded from the anonymity set.
    \item[3.] If the content of exchange is disclosed, the platform knows it is a BTC-to-ETH exchange and excludes David and Ella.  
\end{itemize}

As for the compliance information in the filing operation, someone may consider just leaking the summary of compliance information is fine. 
But in this case, 
many users might have not generated any transactions in a year which can be inferred from their zero profit.
Excluding these sleeping users reduces the anonymity set.
Therefore, the privacy of compliance information should also be protected. 
The identity of the filing operation could be leaked by the regulatory authority to the platform which is an unavoidable leakage. 


In a nutshell, protecting privacy is necessary. We need to go toward the best possible anonymity.  

When modeling the best possible anonymity,  
we want to prevent any avoidable leakage. 
However, since public information could have various and complicated relationships with the events in the exchange system, it is tricky to exactly quantify the potential influence, which may be leveraged by the adversary to win trivially.  
Instead, we define the full anonymity via interaction indistinguishability. 
\noindent{\it Interaction indistinguishability.}
This property requires that the interaction between the user and platform leak nothing except the public information.
To include the interaction of all kinds of transactions, we design the experiment as follows. 
In a high level, the adversary \adv\ acts as the malicious platform and the experiment simulates two worlds with the same initialization. 
\adv\ can add honest users to both worlds and interact with them by submitting different query pairs.
The queries are sent to the worlds via a challenger \cdv\ who forwards query pair to two worlds depending on a random bit \textit{b}.
To model the unavoidable leakage, the queries should contain the same public information. 
After a series of interactions, \adv\ still cannot distinguish which world is based on which one of the query pair. 
It means that for any kind of transaction the interaction does not leak more than the unavoidable leakage. 
Otherwise, \adv\ can send different queries for the transaction and distinguish the two worlds successfully.

The two worlds are simulated via two sets of oracles: 
$\mathcal{O}_\mathsf{IND}^a=\{\mathcal{O}^{1,a}_\mathsf{Join}$, $\mathcal{O}^{1,a}_\mathsf{Deposit}$, $\mathcal{O}^{1,a}_\mathsf{Exchange},\mathcal{O}^{1,a}_\mathsf{Withdraw},\mathcal{O}^{1,a}_\mathsf{File}$, $\mathcal{O}^a_\mathsf{Public}\}$ 
with separated internal map $\textsc{Map}^a$
for $a\in\{0,1\}$. 
\cdv\ chooses one bit \textit{b} randomly at the beginning.
\adv~sends queries to the challenger $\mathcal{C}$ which are in pair $(Q^0, Q^1)$ to interact with oracles.
For each query pair, $\mathcal{C}$ checks that they could be different but must contain the same public information which represents the unavoidable leakage as we discussed before (see Def.~\ref{defquery}).
Then \cdv\ forwards $Q^b$ to the oracle in $\mathcal{O}^0_{\mathsf{IND}}$ and forwards $Q^{1-b}$ to the oracle in $\mathcal{O}^1_{\mathsf{IND}}$.  

With these queries as input, these oracles interact with \adv\ with different states, 
and we denote it in terms of $\langle\adv(st^0), \mathcal{O}^0_{\mathsf{IND}}(Q^b)\rangle$ and $\langle\adv(st^1), \mathcal{O}^1_{\mathsf{IND}}(Q^{1-b})\rangle$. 
But it cannot distinguish which are induced by which queries. Therefore, the interaction leaks nothing but public information.
We formally define the interactive indistinguishability in Fig~\ref{defind}.

\begin{definition}[Publicly consistent queries]\label{defquery}
$\mathcal{A}$ submits a publicly consistent query pair $(Q^0, Q^1)$, which satisfy all the following conditions:
\vspace{-0.2cm}
\setlist{nolistsep}
\begin{itemize}[noitemsep,leftmargin=*]    
\item First of all, both queries would succeed, and are for the same type of oracle.
    \item For queries to $\mathcal{O}^1_{\mathsf{Join}}$, with the same the request info $\mathit{req_{joi}}$ and they get the same user identifier $uid$ as output. 
    \item For queries to $\mathcal{O}^1_{\mathsf{File}}$, both with the same user identifier $uid$.
    \item For queries to $\mathcal{O}^1_{\mathsf{Deposit}}$ and $\mathcal{O}^1_{\mathsf{Withdraw}}$, the users can be different but the name and amount of the assets and the on-chain addresses are the same in both queries. For fiat money deposit/withdraw, the users and bank accounts are the same.
\end{itemize}

\end{definition}

\begin{figure}[htbp]
    \centering
    \fbox{
	\procedure{$\mathrm{Exp}^{\mathsf{IND}}(\mathcal{A},\mathcal{C},\lambda)$}
        {\textit{epp}\leftarrow\mathsf{Setup}(\mathcal{G}(1^{\lambda}))\\
	(pk,st)\leftarrow\mathcal{A}(\textit{epp})\\	
        \mathcal{C}\ \text{randomly chooses } b\sample\{0,1\}\\
        \text{Run } \mathcal{A}^{\mathcal{C}(\mathcal{O}^0_\mathsf{IND},\mathcal{O}^1_\mathsf{IND})}(st)\ \text{for \textit{N} steps:}\ \color{brown}{\text{// \textit{N}=poly($\lambda$)}} \\
        \text{In each step:}\\
        \quad (Q^0, Q^1, st^0, st^1)\leftarrow \mathcal{A}(st)\\
        \quad \pcif (Q^0, Q^1) \text{ are not \textit{publicly consistent}, } \pcthen \pcreturn 0;\\
        \quad \pcelse \mathcal{C} \text{ forwards } Q^b\text{ to } \mathcal{O}^0_\mathsf{IND},\ Q^{1-b}\text{ to } \mathcal{O}^1_\mathsf{IND},\\ 
        \quad \text{Run}\ \langle\adv(st^0),\mathcal{O}^0_\mathsf{IND}(Q^b)\rangle \text{ and } \langle\adv(st^1),\mathcal{O}^1_\mathsf{IND}(Q^{1-b})\rangle\\
        \quad\ \color{brown}{\text{// It simulates \adv\ induces honest users’ behaviors}}\\
        \text{Finally, } \adv\ \text{halts, and outputs } \hat{b}\\
		\pcreturn (\hat{b}==b)
		}
		}
	\caption{Interaction indistinguishability experiment}
	\label{IND}
\end{figure}

\begin{definition}[Interaction indistinguishability]\label{defind} 
The interaction indistinguishability is described in Fig~\ref{IND}. 
We say that an exchange system provides interaction indistinguishability if for all PPT $\mathcal{A}$ and $\lambda$ it holds that
\begin{equation*}
	\left|\mathrm{Pr}[\mathrm{Exp^{\mathsf{IND}}}(\mathcal{A},\lambda)=1]-{1}/{2}\right|\leq negl(\lambda)
\end{equation*}

\end{definition}

\vspace{-0.1cm}	
\begin{remark}[Relation with basic anonymity]
The interaction indistinguishability implies the basic anonymity if the exchange identity, and exchange content are public and identical in both queries. 
Only the withdraw identity is concealed like in the basic withdraw anonymity experiment. 


\end{remark}

\vspace{-0.2cm}	
\begin{remark}[Best possible anonymity]


We claim that the interaction indistinguishability achieves the best of possible anonymity. 
The interaction indistinguishability covers all kinds of transactions with specific public information. 
When we specify that the public information exclusively comprises the unavoidable leakage as defined in Def~\ref{defquery},
we can ensure that the platform \textbf{learns nothing} about the user from their interactions, except for the unavoidable leakage. 
Recall the Example 2, we can see any avoidable leakage in these cases excludes some users. 
If all avoidable leakages are prevented, the anonymity set expands to encompass a broader range of users, now including both David and Ella.


\end{remark}

%% file: 5full.tex
\section{Private and Compliable Exchange System}
\label{construction_system}

In this section, we present the generic construction of the private and compliable exchange system $\Pi_{\mathsf{\name}}$ that achieves full anonymity, overdraft prevention, and compliance, and we provide formal proofs of its security.
Before that, we give concrete compliance rules that our system aims to comply with for both clients and the platform. 
Following that, we introduce the concept of price credentials and illustrate the high-level idea about the construction.
%

\noindent{\em Concrete compliance rules.} \label{compliance-rule}
For F-client-compliance, we take the tax report as an example of F, called tax-report-client-compliance. 
It requires clients to report the investment profit  yearly (total gain minus total cost). 
The taxable profit is calculated when clients sell their assets via exchange or withdraw transactions.
Concretely, $\textit{cost}=\textit{pr}_i\cdot k_i$ and $\textit{gain}=\overline{\textit{pr}}_i\cdot k_i$, where $k_i$ is the exchange-out or withdrawn asset amount, $\textit{pr}_i$ is the buying price and $\overline{\textit{pr}}_i$ is the selling price of the asset. 
Then he reports the accumulated cost $\mathit{cp}_1=\sum \mathit{pr_i}\cdot\mathit{k_i}$ and gain $\mathit{cp}_2=\sum \overline{pr}_i\cdot\mathit{k_i}$ to the authority.





For G-platform-compliance, we refer to the liquidity coverage ratio (LCR) requirement of Basel Accords \cite{Basel}, a series of banking regulations established by representatives from major global financial centers. 
LCR mandates that banks hold sufficient cash and liquid assets to cover fund outflows for 30 days. 
The platform always knows clearly the inflows/outflows for each coin including fiat money transfers (as the platform receives or transfers them out) by checking its internal state.
It can prepare enough amount of coins for all kinds of assets.
Thus this LCR-platform-compliance is compatible with our fully anonymous setting.

\noindent\textit{{About price credential and price fluctuation.}}\label{Pricecredential} 
To achieve efficient private exchange with compliance, we introduce {price credentials} denoted as $\mathit{px:=(time,name,pr,sig)}$,
where the signature $\mathit{px.sig}$ is signed by the platform on the current time $\mathit{px.time}$, the coin name $\mathit{px.name}$, and the corresponding current price $\mathit{px.pr}$. 
To tackle price fluctuation without leaking coin information, the platform keeps signing the latest prices for all coins at the same timestamp.
In the exchange transaction, the user proves that the newly exchanged asset record contains the name and price, and they know a valid signature on these values from the latest timestamp's price credential. The user also ensures that the exchange is fair based on these prices and amounts. 
This approach enables 
the user to prove with just a single credential. 
It significantly reduces communication and computation costs to a constant level. 



\noindent\textit{{High-level idea.}}
Before giving the formal algorithms, let us illustrate the high-level construction idea with five concrete transaction examples as follows. 
The platform only knows some public information, like all registered users and their bank accounts, the name and amount of deposited and withdrawn assets, as shown in Fig~\ref{view}. 


\begin{figure}[!htbp] 
\centering 
\includegraphics[width=1.05\linewidth, height= 0.35\linewidth]{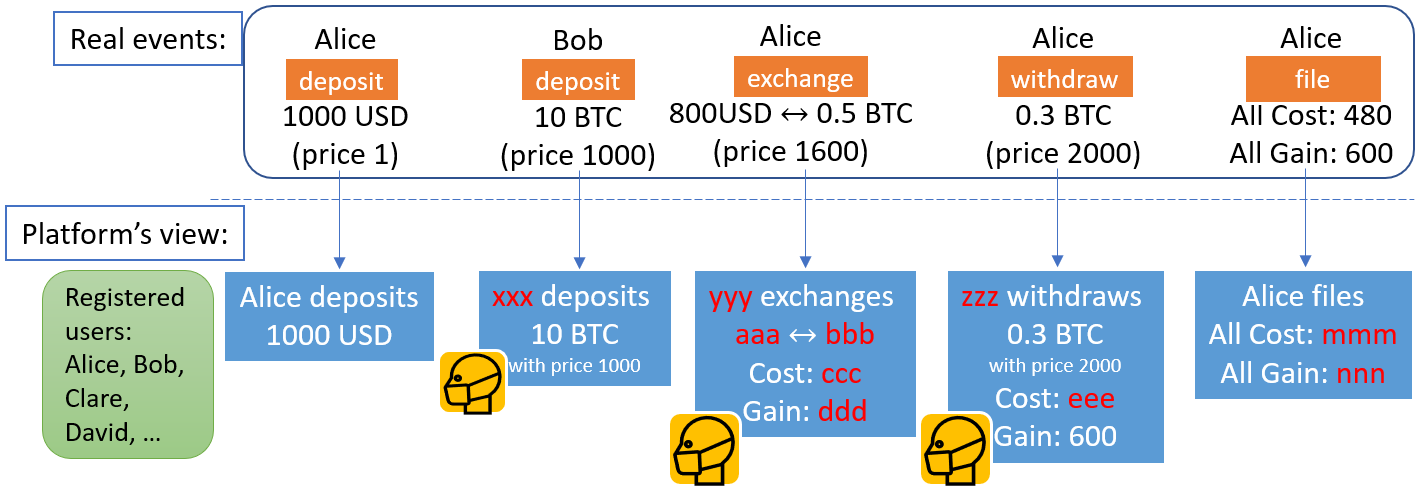} 
\caption{Platform's view in the exchange system}
\label{view} 
\end{figure}

\setlist{nolistsep}
\begin{itemize}[noitemsep,leftmargin=*]
    \item When deposits 1000 USD, the bank account reveals Alice's identity. Then Alice gets an asset credential or signature generated by the platform.  
    To prevent double-spending and ensure regulatory compliance, the signed message must contain additional attributes beyond asset details.
    These attributes include a unique asset identifier, Alice's user identifier. 
    They should be hidden from the platform to keep the user anonymous. 
    Alice is required to prove in zero-knowledge that the blinded user identifier is equal to that in her registration credential. 
    To meet tax-report-client-compliance, the profit of selling assets in exchange and withdraw transactions need to be computed. 
    This necessitates including the exact cost of the assets when they were purchased in deposit and exchange transactions. To this end, we introduce the buying price as an additional attribute in asset credentials. 
    \item When Bob deposits 10 BTC, the only difference from a fiat deposit is that the platform does not know Bob's identity. 
    
    \item When Alice exchanges 800 USD for 0.5 BTC, she uses a 1000 USD asset credential to request two new credentials for the remaining 200 USD and 0.5 BTC, keeping their details hidden. The platform grants her request under specific conditions, including demonstrating in zero-knowledge that she has enough USD, ensuring the credentials share the same user identifier, confirming the non-negative remaining USD amount, matching the BTC's price with the latest credential, and verifying the total exchange value equivalence.
    
    
    We give each asset one separate asset credential for practicality especially with compliance. Intuitively, if all assets are in one credential their attributes would increase linearly with the asset names. 
    Compliance makes it more complicated, because every asset transaction would have to be recorded as an attribute in the credential. This means that the credential attributes would keep growing. 
    It is practical to separate each transaction asset, but it is hard to collect all transactions to calculate the total profit. 
    We solve it by accumulating profit for each transaction involving exchanged-out or withdrawn assets, and recording it on user's exclusive registration credential as two attributes: accumulated cost and accumulated gain. 
    Concretely, Alice shows her registration credential, and requests the platform to issue a new registration credential on a new index, updated cost and gain which are consistent with real cost (amount times buying price) and gain (amount times selling price). 
    \item When withdraws 0.3 BTC, it is similar to the exchange operation, except for the exchanged-in asset. Alice also verifies the receipt of the withdrawn BTC on the blockchain.
    \item When files all cost 480 and all gain 600,  Alice shows a valid registration credential with her identity, and requests an updated registration credential with a new index, reset cost and gain as zeros, and a file credential used to show to the regulatory authority. The file credential contains Alice's real identity, the correct cost 480 and gain 600, and some regulatory auxiliary information. Since the cost and gain are hidden from the platform to avoid information leakage, Alice should prove the committed cost and gain are equal to the ones in her registration credential, and the platform signs blindly. Alice unblinds it and submits the message signature pair to the authority for tax report.
\end{itemize}


\vspace{-0.2cm}	
\subsection{An efficient $\mathsf{\name}$\ construction}
\label{construction}
In this section, we give an efficient $\mathsf{\name}$ construction from additive homomorphic commitment, blind signature and zero-knowledge proof\footnote{Note that these primitives are also used for updatable anonymous credentials and an incentive system in \cite{Bl2019UpdatableSystems}. As we explain in detail in App.\ref{IncentiveSystem}, they do not support the exchange operation and compliance rule in our setting.}. 
Let $\mathsf{Com}$ be an additively homomorphic commitment scheme, $\Pi_{\mathrm{bs}}=(\mathsf{KeyGen}$, $\mathsf{Com}$, $ \langle\mathsf{BlindSign}, \mathsf{BlindRcv}\rangle,\mathsf{Vrfy})$ be a blind signature scheme using $\mathsf{Com}$ to blind messages, and $\mathsf{ZKAoK}$ is the underlying proof system. 
The platform maintains the registered user set $\mathsf{USet}$ and the identifier set $\mathsf{ID}$ which are initially empty.
Formal construction is in Fig~\ref{scheme}. 
The concrete instantiation is presented in section~\ref{instantiaction}. 

\begin{figure*}[!htbp] 
\centering 
\includegraphics[page=2, trim = 18mm 27mm 13mm 29mm, width=1\linewidth,frame]{pages.pdf} 
\caption{Our efficient construction of Pisces}
\label{scheme} 
\end{figure*}

\noindent{\bf Extended compliance support.}\ 
Our construction also easily supports other regulation policies. 
We just give sketches here due to the page limitation.
For example, AML requires that users cannot exchange or withdraw too many times in a time period.
It can be achieved by adding a counter in the registration record.
In each exchange or withdraw transaction, the user proves that the counter in his latest registration record is smaller than some value and the counter is credited by one in the newly issued record. 
Other rules are similar, such as transaction amounts, and (total) value of exchanged assets.

We can also enforce tax filing by prohibiting users who have not filed tax last year from exchanging and withdrawing. 
It can be achieved by adding a year number in the registration record indicating the year when the user filed tax last time.
In each exchange or withdraw transaction, the user shows the year number in his latest registration record and this number is credited by one when the user has filed his tax.

\vspace{-0.3cm}	

\subsection{Security analysis}



\begin{theorem}[Interaction indistinguishability]
    {If $\Pi_{\mathrm{bs}}$ has blindness and the underlying $\mathsf{ZKAoK}$ is zero-knowledge, the commitment is hiding, then $\Pi_{\mathsf{\name}}$ has interaction indistinguishability.}
\end{theorem}
\vspace{-0.3cm}

\begin{proof}
We prove this theorem by a sequence of hybrid experiments $(\mathsf{G}_{\mathrm{real}}, \mathsf{G}_{\mathrm{1}}, \mathsf{G}_{\mathrm{sim}})$.
$\mathsf{G}_{\mathrm{real}}$ is the original $\mathsf{IND}$ experiment.
$\mathsf{G}_{1}$ modifies $\mathsf{G}_{\mathrm{real}}$ by simulating the $\mathsf{ZKAoK}$ proof.
$\mathsf{G}_{\mathrm{sim}}$ modifies $\mathsf{G}_{1}$ by replacing the original commitments with commitments on random strings. 
Since the underlying $\mathsf{ZKAoK}$ is zero-knowledge, $\mathsf{G}_{1}$ can be distinguished from $\mathsf{G}_{\mathrm{real}}$ with only negligible probability.
Due to that the commitment scheme is hiding and the blind signature has blindness, $\mathsf{G}_{\mathrm{sim}}$ can be distinguished from $\mathsf{G}_{1}$ with only negligible probability.
Thus $\mathsf{G}_{\mathrm{sim}}$ can be distinguished from $\mathsf{G}_{\mathrm{real}}$ with only negligible probability.
Furthermore, in $\mathsf{G}_{\mathrm{sim}}$, \adv's view is fully simulated, independent of $b$, \adv's advantage in $\mathsf{G}_{\mathrm{sim}}$ is 0. So \adv\ wins in $\mathsf{G}_{\mathrm{real}}$ with at most negligible probability.  
We describe $\mathsf{G}_{\mathrm{1}}, \mathsf{G}_{\mathrm{sim}}$ as follows.


\noindent $\mathsf{G}_{1}$:
This experiment modifies $\mathsf{G}_{\mathrm{real}}$ by simulating the $\mathsf{ZKAoK}$ proof. 
It works as follows:
at the beginning, 
$\mathcal{C}$ chooses $b\sample\{0,1\}$ and generates $pp\leftarrow\mathsf{Setup}(1^{\lambda})$ and zero-knowledge trapdoor $\textit{td}\leftarrow\mathsf{Sim}(1^{\lambda})$.
$\mathcal{C}$ sends $pp$ to $\mathcal{A}$ and initializes two sets of oracles $\mathcal{O}^0_{\mathsf{IND}}$ and $\mathcal{O}^1_{\mathsf{IND}}$.
In the following oracle queries, the proofs are generated by $\mathcal{C}$ using $\textit{td}$: $\pi\leftarrow\mathsf{Sim}(\textit{td}, x)$.
$\mathsf{G}_{1}$ proceeds in steps, and each time $\mathcal{A}$ queries an oracle, it sends $\mathcal{C}$ a pair of queries $(Q^0, Q^1)$.
$\mathcal{C}$ first checks that they are \textit{publicly consistent} according to Def. \ref{defquery}, then simulates different oracles.

\setlist{nolistsep}
\begin{itemize}[noitemsep,leftmargin=*]

\item[-] For $\mathcal{O}^1_\mathsf{Join}$ oracle, $\mathit{Q^0=(req_{joi}^0, ref_{reg}^0)}$ and $\mathit{Q^1=(req_{joi}^1, ref_{reg}^1)}$ . 
To answer them, $\mathcal{C}$ behaves as in $\mathsf{G}_{\mathrm{real}}$ except for the following modification.
The $\mathsf{ZKAoK}$ proofs $\pi^0, \pi^1$ are simulated using $\textit{td}$.
$\mathcal{C}$ replies $\mathcal{A}$ with $(\textit{com}^b, \pi^0)$ and $(\textit{com}^{1-b},\pi^1)$. 
If $\mathcal{A}$ accepts the proofs, it runs $\hat{\sigma}^0\leftarrow\mathsf{BlindSign}(pp,pk,com^0)$ and $\hat{\sigma}^{1}\leftarrow\mathsf{BlindSign}(pp,pk,com^{1})$ and sends them to $\mathcal{C}$ and continues.

\item[-] For $\mathcal{O}^1_{\mathsf{Deposit}}$ oracle, $\mathit{Q^0=(uid^0, req_{dep}^0,}\mathit{ref}^0_\mathsf{reg}, \mathit{ref}_\mathsf{ast}^{\mathsf{out}1})$ and $\mathit{Q^1=(uid^1, req_{dep}^1,} \mathit{ref}^1_\mathsf{reg}, \mathit{ref}_\mathsf{ast}^{\mathsf{in}1},$ $ \mathit{ref}_\mathsf{ast}^{\mathsf{out}1})$. 
To answer them, $\mathcal{C}$ behaves as in $\mathsf{G}_{\mathrm{real}}$ except the following modification:
It simulates the $\mathsf{ZKAoK}$ proofs $\pi^0, \pi^1$ using $\textit{td}$. 
$\mathcal{C}$ replies $\mathcal{A}$ with $(\{\mathit{com}_u^b\}_{u=1}^2, \pi^0)$ and $(\{\mathit{com}_u^{1-b}\}_{u=1}^2, \pi^1)$. 
If $\mathcal{A}$ accepts the proofs, it runs the $\mathsf{BlindSign}$ algorithm and sends the respective blinded signatures to $\mathcal{C}$ and continues.


\item[-] For $\mathcal{O}^1_{\mathsf{Exchange}}$ oracle, $\mathit{Q^0=(uid^0,req^0_{exc},}\mathit{ref}^0_{\mathsf{reg}}, \mathit{ref}_{\mathsf{ast}}^{\mathsf{in}0},\ $ $ \mathit{ref}_{\mathsf{ast}}^{\mathsf{out}0})$ and $\mathit{Q^1=(uid^1,req^1_{exc},}\mathit{ref}^1_{\mathsf{reg}}, \mathit{ref}_{\mathsf{ast}}^{\mathsf{in}1}, \mathit{ref}_{\mathsf{ast}}^{\mathsf{out}1})$.  
To answer them, $\mathcal{C}$ behaves as in $\mathsf{G}_{\mathrm{real}}$ except that 
it 
simulates the $\mathsf{ZKAoK}$ proofs $\pi^0, \pi^1$ using $\textit{td}$. 
$\mathcal{C}$ replies $\mathcal{A}$ with $(\{\mathit{com}_u^b\}_{u=1}^7, \pi^0)$ and $(\{\mathit{com}_u^{1-b}\}_{u=1}^7, \pi^1)$. 
If $\mathcal{A}$ accepts the proofs, it runs the $\mathsf{BlindSign}$ algorithm and sends the blinded signatures to $\mathcal{C}$ and continues.
\item[-] For $\mathcal{O}^1_{\mathsf{Withdraw}}$ oracle, $\mathit{Q^0 =(uid^0,req^0_{wit},}\mathit{ref}^0_{\mathsf{reg}}, \mathit{ref}_{\mathsf{ast}}^{\mathsf{in}0})$ and $\mathit{Q^1 =(uid^1,req^1_{wit},}\mathit{ref}^1_{\mathsf{reg}}, \mathit{ref}_{\mathsf{ast}}^{\mathsf{in}1})$. 
To answer them, $\mathcal{C}$ behaves as in $\mathsf{G}_{\mathrm{real}}$ except that
it simulates the $\mathsf{ZKAoK}$ proofs $\pi^0,\pi^1$ using $\textit{td}$. 
$\mathcal{C}$ replies $\mathcal{A}$ with $(\{\mathit{com}_u^b\}_{u=1}^4, \pi^0)$ and $(\{\mathit{com}_u^{1-b}\}_{u=1}^4, \pi^1)$. 
If $\mathcal{A}$ accepts the proofs, it runs the $\mathsf{BlindSign}$ algorithm and sends the blinded signatures to $\mathcal{C}$ and continues.  
\item[-] For $\mathcal{O}^1_{\mathsf{File}}$ oracle, $\mathit{Q^0 = (uid^0, req^0_{fil}},\mathit{ref}^0_{\mathsf{reg}})$, and $\mathit{Q^1 = (uid^1, req^1_{fil}},\mathit{ref}^1_{\mathsf{reg}})$. 
To answer them, $\mathcal{C}$ behaves as in $\mathsf{G}_{\mathrm{real}}$ except that 
it simulates the $\mathsf{ZKAoK}$ proofs $\pi^0, \pi^1$ using $\textit{td}$. 
$\mathcal{C}$ replies $\mathcal{A}$ with 
$(\{\mathit{com}_u^b\}_{u=1}^3, \pi^0)$ and $(\{\mathit{com}_u^{1-b}\}_{u=1}^3, \pi^1)$. 
If $\mathcal{A}$ accepts the proofs, it runs the $\mathsf{BlindSign}$ algorithm and sends the blinded signatures to $\mathcal{C}$ and continues.  
   
\end{itemize}

Note that from $\mathsf{G}_{\mathrm{real}}$  to $\mathsf{G}_{1}$, the only difference is that the $\mathsf{ZKAoK}$ proofs are simulated. 
Due to that the $\mathsf{ZKAoK}$ scheme is zero-knowledge, we have that 
$|\mathrm{Pr}[\mathsf{G}_{\mathrm{real}}(\adv,\lambda)=1]-\mathrm{Pr}[\mathsf{G}_{1}(\adv,\lambda)=1]|\leq negl(\lambda)$.

\noindent$\mathsf{G}_{\mathrm{sim}}$: 
This experiment modifies $\mathsf{G}_{1}$ by replacing the original commitments with commitments on random strings. 
$\mathsf{G}_{\mathrm{sim}}$ proceeds in steps, 
and each time $\mathcal{A}$ invokes an oracle, it sends $\mathcal{C}$ a pair of queries $(Q^0, Q^1)$.
$\mathcal{C}$ first checks that they are \textit{publicly consistent}, 
then simulates different oracles as follows.

\vspace{-0.2cm}
\setlist{nolistsep}
\begin{itemize}[noitemsep,leftmargin=*]
\item[-] For $\mathcal{O}^1_\mathsf{Join}$ oracle, 
to answer queries 
$Q^0, Q^1$, $\mathcal{C}$ behaves as in $\mathsf{G}_{1}$ except that it produces commitment $\textit{com}^0, \textit{com}^1$ on random strings $r^0, r^1$. 
\item[-] For $\mathcal{O}^1_{\mathsf{Deposit}}$ oracle, 
to answer $Q^0, Q^1$, $\mathcal{C}$ behaves as in $\mathsf{G}_1$ except that it produces commitments $\{\mathit{com}_u^0\}_{u=1}^2$ and $\{\mathit{com}_u^{1}\}_{u=1}^2$ on random strings.

\item[-] For $\mathcal{O}^1_{\mathsf{Exchange}}$ oracle, 
to answer $Q^0, Q^1$, $\mathcal{C}$ behaves as in $\mathsf{G}_1$ except that it produces commitments $\{\mathit{com}_u^0\}_{u=1}^7$ and $\{\mathit{com}_u^{1}\}_{u=1}^7$ on random strings. 

\item[-] For $\mathcal{O}^1_{\mathsf{Withdraw}}$ oracle, 
to answer $Q^0, Q^1$, $\mathcal{C}$ behaves as in $\mathsf{G}_1$ except that it produces commitments $\{\mathit{com}_u^0\}_{u=1}^4$ and $\{\mathit{com}_u^{1}\}_{u=1}^4$ on random strings. 

\item[-] For $\mathcal{O}^1_{\mathsf{File}}$ oracle, 
to answer $Q^0, Q^1$, $\mathcal{C}$ behaves as in $\mathsf{G}_1$ except that it produces commitments $\{\mathit{com}_u^0\}_{u=1}^3$ and $\{\mathit{com}_u^{1}\}_{u=1}^3$ on random strings. 

\end{itemize}

In each of the above cases, $\mathcal{A}$'s view is independent of $b$. 
Thus, $\mathcal{A}$ just outputs a random guess $\hat{b}$ in $\mathsf{G}_{\mathrm{sim}}$, so its advantage is 0: $\mathrm{Pr}[\mathsf{G}_{\mathrm{sim}}(\adv,\lambda)=1]-1/2=0$

Note that from $\mathsf{G}_{1}$ to $\mathsf{G}_{\mathrm{sim}}$, we change that the commitments are on the random strings. 
Due to the hiding property of the commitment scheme and the blindness of $\Pi_{\mathrm{bs}}$ (which is also based on the hiding property of the commitment), 
we have that 
$|\mathrm{Pr}[\mathsf{G}_{1}(\adv,\lambda)=1]-\mathrm{Pr}[\mathsf{G}_{\mathrm{sim}}(\adv,\lambda)=1]|\leq negl(\lambda)$.
In summary, we have that
\begin{equation*}
    \begin{aligned}
    &|\mathrm{Pr}[\mathrm{Exp^{\mathsf{IND}}}(\mathcal{A},\lambda)=1]-{1}/{2}| 
    =|\mathrm{Pr}[\mathsf{G}_{\mathrm{real}}(\adv,\lambda)=1]-{1}/{2}|\\ 
    \leq&|\mathrm{Pr}[\mathsf{G}_{\mathrm{real}}(\adv,\lambda)=1]-\mathrm{Pr}[\mathsf{G}_{1}(\adv,\lambda)=1]|\\ &+|\mathrm{Pr}[\mathsf{G}_{1}(\adv,\lambda)=1]-\mathrm{Pr}[\mathsf{G}_{\mathrm{sim}}(\adv,\lambda)=1]|\\
    &+|\mathrm{Pr}[\mathsf{G}_{\mathrm{sim}}(\adv,\lambda)=1]-1/2|\\
    \leq& negl(\lambda)    
    \end{aligned}
\end{equation*}

\vspace{-0.7cm}
\end{proof}

\begin{theorem}[Overdraft prevention]
    If the underlying $\mathsf{ZKAoK}$ has argument of knowledge, the commitment is binding and  $\Pi_{\mathrm{bs}}$ is unforgeable, then $\Pi_{\mathsf{\name}}$ has overdraft prevention.
\end{theorem}

\vspace{-0.3cm}

\begin{proof}
In the overdraft prevention experiment, the adversary $\mathcal{A}$ wins if it withdraws more asset than it has deposited or exchanged. 
Given the transaction histories $\mathit{h_t=(uid,}\mathit{Rd_{reg}}, \mathit{Rd_{ast}},\mathit{Rd'_{reg}}, \mathit{Rd'_{ast}}, \mathit{Rd_{ast}^{out}},\mathit{ts_t},\mathit{pub_t})$ 
extracted by \edv\ for $t\in[N]$. 
If \adv\ wins, there exists at least one  transaction has issues which means the input or output records in it are problematic such that one of the following events happens:  
\vspace{-0.2cm}

\setlist{nolistsep}
\begin{itemize}[noitemsep,leftmargin=*]
    \item [1.] The record used in this transaction is not generated from previous transactions, i.e., $Rd\notin \mathsf{RdSet}$; 
    \item [2.] The user steals other honest users' assets;
    \item[3.] The generation of asset record is wrong in one of the following cases: \\
{Deposit}: $\textit{Rd}_{\textit{ast}}^{\textit{out}}.\textit{{name}} \neq \textit{ts}_{\textit{t}}.\textit{name}$ or $\textit{Rd}_{\textit{ast}}^{\textit{out}}.\textit{{amt}} \neq \textit{ts}_{\textit{t}}.\textit{{amt}}$;\\
{Exchange}: $\mathit{Rd'_{ast}}.\textit{{name}} \neq \mathit{Rd_{ast}}.\textit{name}$ {or} $\mathit{Rd'_{ast}}.\textit{{amt}} <0$ or $(\mathit{Rd_{ast}}.\textit{amt}-\mathit{Rd'_{ast}}.\textit{amt})\cdot \textit{pub}_{\mathit{t}}.\mathit{pr_{in}} \neq \mathit{Rd_{ast}^{out}}.\mathit{amt}\cdot \textit{pub}_{\textit{t}}.\mathit{pr}_{out}$;\\
{Withdraw}: $\mathit{Rd_{ast}}.\textit{{name}} \neq \mathit{ts_{t}}.\textit{name} \ \mathrm{or}\ \mathit{Rd'_{ast}}.\textit{{amt}} <0$ or  $\mathit{Rd_{ast}.\textit{amt}}-\mathit{Rd'_{\textit{ast}}.\textit{amt}}<\mathit{ts_{t}.\textit{amt}}$. 
\end{itemize}

For events 1 and 2, the input records are problematic which are forged or stole by \adv. 
\adv\ may forge the asset records or reuse records with a different identifier. 
In order to use others' asset, \adv\ must guess the \textit{aid} correctly. 
For event 3, the new generated asset records are problematic, \adv\ gets these records by cheating the issuer.
The security of our scheme can be reduced to the underlying cryptographic building blocks. 
This includes standard primitives like commitment, blind signature, and non-interactive $\mathsf{ZKAoK}$.
We elaborate it case by case. 
\vspace{-0.2cm}

\noindent(1) 
Suppose that Pr[$\mathcal{A}$ wins and event 1 happens] is non-negligible. In this case, it leads to at least one of the following contradictions: 
\vspace{-0.2cm}
\setlist{nolistsep}
\begin{itemize}[noitemsep,leftmargin=*]
    \item The record $Rd$ is valid but was not generated via querying oracles, which breaks the unforgeability of $\Pi_mathrm{bs}$;
    \item The asset record  $Rd$ is reused with another identifier. In this case, since the used identifier would be detected, the revealed identifiers must be different $aid \neq aid'$. 
    It means one commitment produces two different openings which contradicts the binding property of the commitment. 
\end{itemize}

\vspace{-0.2cm}
\noindent(2) 
Suppose that Pr[$\mathcal{A}$ wins and event 2 happens] is non-negligible. 
Here the input records are valid records generated from previous transaction but belong to other honest users. 
If \adv\ uses this asset record, it must know the respective \textit{aid} which is kept privately by the honest user. 
It contradicts that $\mathcal{A}$ can only guess it correctly with negligible probability. 


\vspace{-0.2cm}
\noindent(3) 
Suppose that Pr[$\mathcal{A}$ wins and event 3 happens] is non-negligible. 
In this case, \adv\ generates a valid transaction but gets asset records with wrong attributes.
It leads to at least one of the following contradictions:
\vspace{-0.2cm}
\setlist{nolistsep}
\begin{itemize}[noitemsep,leftmargin=*]
    \item For deposit transaction, \adv\ gets an asset record which is different from the deposit request. It happens only if one commitment produces two different openings which contradicts the binding property or $\mathcal{A}$ uses the incorrect witness to generate a valid proof, which breaks the argument of knowledge of underlying $\mathsf{ZKAoK}$.
    \item For exchange transaction, \adv\ gets new exchange-out asset record which is different from the old one or gets exchange-in asset with more amount by breaking the fair exchange rule. 
    It leads to at least one of the following contradictions:
\vspace{-0.2cm}

\setlist{nolistsep}
\begin{itemize}[noitemsep,leftmargin=*]
    \item[-] for the commitments of records, $\mathcal{A}$ generates a valid proof with incorrect witness, which breaks the argument of knowledge of underlying $\mathsf{ZKAoK}$; 
    \item[-] $\mathcal{A}$ opens the commitment to different values and generates the proof. It means one commitment produces two different openings which contradicts the binding property of the commitment scheme;
    \item[-] the price credentials are forged by \adv, thus the platform's signatures. It contradicts to the unforgeability of $\Pi_{\mathrm{bs}}$.
\end{itemize} 
    \item For withdraw transaction, the user should prove that it owns enough asset for the withdraw request by committing on the old asset and new asset. Then he proves that the opening of the asset name is the same as that in the request and the deducted amount is the same as the withdrawal amount and the new amount is non-negative.   
    Now the new asset record does not meet at least one of these requirements. 
    It happens only if one commitment produces two different openings which contradicts the binding property or $\mathcal{A}$ uses the incorrect witness to generate a valid proof, which breaks the argument of knowledge of underlying $\mathsf{ZKAoK}$.
\end{itemize}

\vspace{-0.4cm}
\end{proof}

\begin{theorem}[Compliance]\label{thmcomp}
    If the underlying $\mathsf{ZKAoK}$ is secure with argument of knowledge, the commitment is binding, and $\Pi_{\mathrm{bs}}$ is unforgeable, then $\Pi_{\mathsf{\name}}$ has 
     tax-report-client-compliance.
\end{theorem}
\vspace{-0.3cm}	

\begin{proof}
We prove the tax-report-client-compliance as follows.
In this experiment, $\mathcal{A}$ wins if it outputs $\mathit{doc}$ which passes the authority verification but is inconsistent with the transaction histories. Given that 
 \edv\ extracts transaction histories $\mathit{h_t=(uid,}\mathit{Rd_{reg}},\mathit{Rd_{ast}},$ $ \mathit{Rd'_{reg}}, \mathit{Rd'_{ast}}, \mathit{Rd_{ast}^{out}},\mathit{ts_t},\mathit{pub_t})$ for $t\in[N]$. 
$\mathcal{A}$ wins if one of the following events happens:
\vspace{-0.2cm}
\setlist{nolistsep}
\begin{itemize}[noitemsep,leftmargin=*]
    \item [1.] $\mathit{doc}$ was not obtained from queries but passed the authority verification.
    \item [2.] The user uses others' registration record: the user identities of asset and registration records are not the same;
    \item[3.] The price was inconsistent as follows:
    In deposit transaction, $\textit{Rd}_{\textit{ast}}^{\textit{out}}.\textit{acp}\neq\textit{pub}_{\textit{t}}.pr_{out}$; or in exchange transaction, $\textit{Rd}'_{\textit{ast}}.\textit{acp}\neq\textit{Rd}_{\textit{ast}}.\textit{acp}$ or $\textit{Rd}_{\textit{ast}}^{\textit{out}}.\textit{acp}\neq\textit{pub}_{\textit{t}}.pr_{out}$; Or in withdraw transaction, $\mathit{Rd'_{ast}}.\mathit{acp}\neq\mathit{Rd_{ast}}.\mathit{acp}$.
    \item[4.] $\mathit{doc}$ was obtained by interacting with $\mathcal{O}_{\mathsf{Sign}}$, but the inconsistency happens since the compliance information was updated incorrectly in some exchange or withdraw transaction;
    \item[5.] The user identifier has not been registered : $uid \notin \mathsf{RU}$ in any transaction expect for Join;
    \end{itemize}
\vspace{-0.2cm}
In a high level, event 1 happens meaning that \adv\ forges a valid signature which is contradicted by the unforgeability of $\Pi_{\mathrm{bs}}$. Events 2,3,4, happens meaning that \adv\ finds the collisions of commitment that violate the binding property of commitment, or \adv\ proves on a wrong statement that violates the argument of knowledge property of non-interactive $\mathsf{ZKAoK}$. Event 5 happens which contains three possible cases. The first is \adv\ forges a record with new $uid$, which violates the unforgeability of $\Pi_{\mathrm{bs}}$. The second is \adv\ finds collisions on $uid$, which violates the binding property of commitment. The third is that \adv\ proves a wrong statement including the unregistered $uid$, which violates the argument of knowledge property of non-interactive $\mathsf{ZKAoK}$.
So, we reduce the security of our scheme to the unforgeability of $\Pi_{\mathrm{bs}}$, the binding property of commitment, and the argument of knowledge property of non-interactive $\mathsf{ZKAoK}$.
We elaborate it case by case. 

\vspace{-0.2cm}
\noindent(1) Suppose that Pr[$\mathcal{A}$ wins and event 1 happens] is non-negligible. 
In this case, $\mathcal{A}$ works honestly for each transaction but sends a $\mathit{doc}$ to the authority which contains the incorrect $\textit{cp}_1, \textit{cp}_2$ and a forged signature on them. 
It breaks the unforgeability of the blind signature.

\vspace{-0.2cm}
\noindent(2) Suppose that Pr[$\mathcal{A}$ wins and event 2 happens] is non-negligible. 
In this case, a valid transaction is generated but the records belong to different users. 
However, the user needs to prove that all records belong to himself by proving they contain the same $uid$ which is a contradiction.
So it breaks the argument of knowledge of the underlying $\mathsf{ZKAoK}$.

\vspace{-0.2cm}
\noindent(3) Suppose that Pr[$\mathcal{A}$ wins and event 3 happens] is non-negligible. 
In this case, \adv\ generates a commitment for its new asset containing its $\textit{uid}$, asset identifier \textit{aid}, asset name \textit{i}, amount $\textit{k}_i$ and price $\textit{pr}_i$. 
Here $\textit{pr}_i$ is different from the real price w.r.t the output of $\mathcal{O}_{\mathsf{Public}}$.
For the deposit transaction, the price is different from the public price. For the exchange transaction, the price is different from that of the old asset record or the price credential. For the withdraw transaction, the price is different from that of the old asset record. 
The occurrence of the incorrect price leads to at least one of the following contradictions:
\vspace{-0.2cm}
\setlist{nolistsep}
\begin{itemize}[noitemsep,leftmargin=*]
    \item $\mathcal{A}$ uses the incorrect witness to generate a valid proof, which breaks the argument of knowledge of underlying $\mathsf{ZKAoK}$; 
    \item $\mathcal{A}$ opens the commitment to different values and generates the proof. It means one commitment produces two different openings which contradicts the binding property of the commitment scheme;
    \item In the exchange transaction, \adv\ manipulates the price by using the price credential forged by itself. It breaks the unforgeability of the blind signature scheme $\Pi_{\mathrm{bs}}$. 
\end{itemize}

\vspace{-0.2cm}
\noindent(4) Suppose that Pr[$\mathcal{A}$ wins and event 4 happens] is non-negligible. 
In this case, for at least one exchange or withdraw transaction
the new compliance information $\textit{cp}_1^*, \textit{cp}_2^*$ was incorrect but the proof is valid.
It leads to at least one of the following contradictions:
\setlist{nolistsep}
\begin{itemize}[noitemsep,leftmargin=*]
    \item[$\bullet$] When computing $\textit{cp}_1^*, \textit{cp}_2^*$, $\mathcal{A}$ uses some incorrect selling prices different from the output of $\mathcal{O}_{\mathsf{Public}}$.
    Its success implies that it breaks the argument of knowledge of underlying $\mathsf{ZKAoK}$, or breaks the binding property of the commitment scheme, or forges a price credential (in the exchange transaction) which breaks the unforgeability of the blind signature.
    \item[$\bullet$] When proving the correctness of $\textit{cp}_1^*, \textit{cp}_2^*$, $\mathcal{A}$ just uses the incorrect witness to generate a valid proof, which breaks the argument of knowledge of underlying $\mathsf{ZKAoK}$; 
    \item[$\bullet$] $\mathcal{A}$ opens the commitment to different compliance information values and generates the proof. It means one commitment produces two different openings which contradicts the binding property of the commitment scheme.
\end{itemize}

\vspace{-0.2cm}
\noindent(5) 
Suppose that Pr[$\mathcal{A}$ wins and event 5 happens] is non-negligible. 
In this case, $uid$ has not registered but \adv\ generates a valid transaction on it which leads to at least one of the following contradictions:
\vspace{-0.2cm}
\setlist{nolistsep}
\begin{itemize}[noitemsep,leftmargin=*]
    \item \adv\ forges a registration record in which the $\sigma_{\textit{reg}}$ should be issued by the platform via a blind signature scheme, so \adv\ breaks the unforgeability of the blind signature; 
    \item[$\bullet$] 
    \adv\ does not have the $\sigma_{\textit{reg}}$ but generates a valid proof in the deposit, exchange or withdraw protocol, so it breaks the argument of knowledge of the underlying $\mathsf{ZKAoK}$.

\end{itemize} 
\vspace{-0.5cm}
\end{proof}
\vspace{-0.1cm}


%% file: 6perform.tex
\section{Performance Evaluations}
\label{sec:eval}
In this section, we describe our instantiation, prototype implementation and the performance evaluation. The evaluation results show that our design is efficient and practical. 

\noindent{\bf Instantiation and implementation.}\label{instantiaction} We instantiate the anonymous exchange system using the Pointcheval Sanders blind signatures~\cite{Pointcheval2016ShortSignatures} and Pedersen commitment~\cite{Pedersen91}. The $\mathsf{ZKAoK}$s are instantiated with $\Sigma$-protocol  on the knowledge of DLog, its equality, and range. 
%
We implement this instantiation of the anonymous exchange system with Java. We use the open source Java library $\mathsf{upb.crypto}$\footnote{upb.crypto: \url{https://github.com/upbcuk}.} and the bilinear group provided by $\mathsf{mcl (bn256)}$\footnote{mcl: \url{https://github.com/herumi/mcl}.}. We run experiments on MacBook Air (1.6 GHz Dual-Core Intel Core i5, 16GB memory).

\begin{table}[htp]
\centering
\caption{Avg. computation cost in milliseconds.
}\label{Timecost}
\begin{tabular}{|c|c|c|c|c|}
\hline
Party    & Join & Deposit & Exchange & Withdraw \\ \hline
\name-user     & 9    & 11       & 46       & 37       \\ \hline
\name-platform & 7    & 14      & 88       & 62      \\  \hline
\end{tabular}
\vspace{-0.5cm}
\end{table}
\vspace{-0.15cm}

\noindent{\bf Performance.} We test the pure computation time cost and communication cost of each procedure to show the efficiency. Then to show the practicality, we make two comparisons. One is to compare the secure exchange with plain exchange to show the overhead is truly small. The other is to compare with other anonymous credential applications, including Privacy Pass and the privacy-preserving incentive system (PPIS for short).  

\noindent{\emph{Computation cost.}} We test the computation time cost of each party in each procedure of the anonymous exchange system. 
As shown in table~\ref{Timecost}, we can see that each party's time cost for each procedure is less than $88ms$, which is quite efficient. 

\noindent{\emph{Communication cost.}} We  measure the communication cost of each procedure and none of them exceeds 12kb. Concretely, in the Join and deposit procedures, the user adds \textasciitilde 2.6kb and \textasciitilde 3.3kb data to the request,  respectively. The platform adds a \textasciitilde 1.8kb data to both responses. 
In the exchange and withdraw procedure, the user adds \textasciitilde 12kb and \textasciitilde 8.7kb data to the request, respectively. The platform adds \textasciitilde 2.3kb and \textasciitilde 2.8kb data to the response, respectively.
\begin{figure}[h] 
\centering 
\includegraphics[width=1\linewidth, height= 0.35\linewidth]{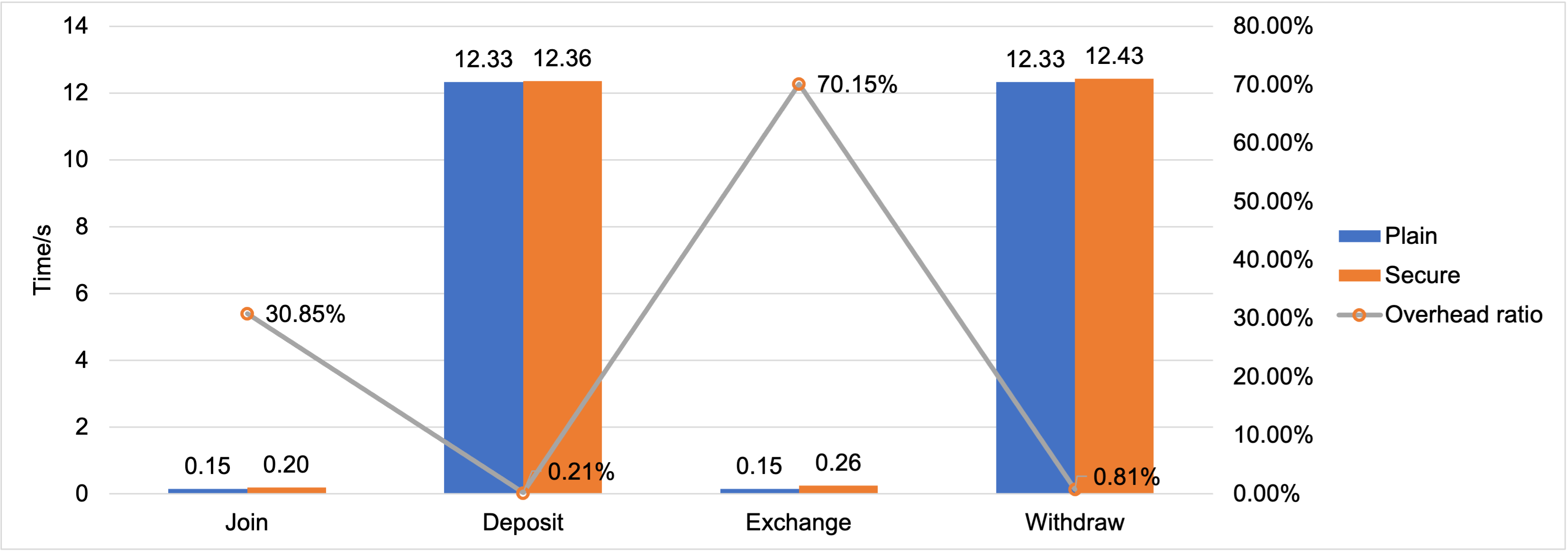} 
\caption{Comparison between plain exchange system and \name}
\vspace{-0.6cm}
\label{fig:comparison} 
\end{figure}

\noindent{\emph{Comparison with plain exchange.}} To demonstrate its practicality, we have taken into consideration the cost of secure communication and have provided a comparison of the estimated time costs between plain operations and secure operations, as shown in Figure~\ref{fig:comparison}. For plain operations, we have estimated the lower-bound time costs by considering only communication cost and on-chain transaction confirmation time, assuming computation cost to be 0. 
We detail the estimation of plain and secure join, deposit, exchange, and withdraw in the following. Consider the optimal network performance, $30-40 ms$ is the desired round-trip time (RTT)\footnote{Network latency:\url{https://www.ir.com/guides/what-is-network-latency}.}. In the estimation, we pick RTT $= 30ms$. 

\vspace{-0.1cm}
\noindent\textit{a) {Join operation}:} For a new user, the plain join includes the sign-up procedure and identity verification for KYC without on-chain confirmation cost. The communication cost includes one TLS handshake with at least 2 round-trip time (RTT for short) cost, 2 RTTs for sign-up setting username and password, and at least 1 RTT for identity verification. Totally the time cost is 5 RTT say $150ms$.
The secure join runs all the plain join process and additionally runs the $\langle\mathsf{Join},\mathsf{Issue}\rangle$ protocol. The time overhead includes 1 RTT for interaction latency, user and platform computation time $16ms$, and data transfer time $\frac{2.6kb}{10MB/s}/ +\frac{1.8kb}{100MB/s} \approx 0.278
ms$. (We assume for a user device the uploading speed is $10MB/s$ and downloading speed is $100MB/s$) The total time cost is $196.278ms$

\vspace{-0.1cm}
\noindent\textit{b) {Deposit operation}:} A plain ETH deposit includes one handshake with platform costing at least 2 RTTs, log-in procedure to get receipt address costing 1 RTT, on-chain payment request costing at least 1 RTT, and an Ethereum transaction confirmation time $12.21s$. The total time cost of the plain deposit is $12.33s$. In an anonymous ETH deposit, users do not log in the platform saving 1 RTT, but run all other procedures of plain deposit. Then users additionally interact with the platform running $\langle \textit{Deposit}, \textit{Credit}\rangle$ where the total computation cost is $25ms$, the interaction with the platform costs 1 RTT, and the data transfer time is $\frac{3.4kb}{10MB/s}/ +\frac{1.8kb}{100MB/s} \approx 0.358ms$. The total time cost of the secure deposit is around $12.355s$.

\vspace{-0.1cm}
\noindent\textit{c) {Exchange operation}:} A plain exchange includes server authentication via TLS handshaking at least 2 RTTs, user login costing 1 RTT, price fetching with 1 RTT, and  sending exchange request with 1 RTT.
The total time cost of a plain exchange is at least 5 RTT, around $150ms$.
The secure exchange removes login, but additionally runs the $\langle \textit{Exchange}, \textit{Update}\rangle$ protocol, where the computation cost is $134ms$, interaction equals the exchange request sending, and data transfer costs $\frac{12kb}{10MB/s}/ +\frac{2.8kb}{100MB/s} \approx 1.228ms$. The total time cost of secure exchange is around $255.228ms$.

\vspace{-0.1cm}
\noindent\textit{d) {Withdraw operation}:} A plain withdraw of ETH includes server authentication via TLS handshaking at least 2 RTTs, user login costing 1 RTT, sending withdraw request with 1 RTT, and waiting the on-chain confirmation with $12.21s$. 
The total time cost of a plain withdraw is around $12.33s$.
The secure exchange gets rid of the login, saving 1 RTT, but additionally requests the price with 1 RTT, and run the $\langle \textit{Withdraw}, \textit{Deduct}\rangle$ protocol, where the computation cost is $99ms$, interaction equals the withdraw request sending, and data transfer costs  $\frac{8.7kb}{10MB/s}/ +\frac{2.3kb}{100MB/s} \approx 0.893ms$. The total time cost of a secure exchange is about $12.43s$.

The results show that the time costs for plain and secure operations are similar, with the overhead of each secure operation being less than $0.11s$. Notably, the overhead ratio of secure deposit and withdrawal is less than $1\%$.

\noindent{\emph{Comparison with Privacy Pass and PPIS.}} To provide a better understanding of the practicality of our system, we conduct performance comparisons with widely used anonymous user-authentication mechanism Privacy Pass~\cite{PrivacyPass}. Privacy Pass published preliminary tests on consumer hardware, indicating that creating a pass in the extension takes less than 40ms\footnote{Privacy Pass FAQ: \url{https://privacypass.github.io/faq/}}. Although the test environments may not be identical to ours, as both are on consumer hardware, the key takeaway is that each procedure of our system incurs similar time costs as Privacy Pass, showcasing its practicality. It's important to note that our system offers additional functionalities beyond Privacy Pass's anonymous authentication. We also test the time cost of the privacy-preserving incentive system (PPIS)~\cite{Bl2019UpdatableSystems}. 
The results, as shown in table~\ref{TimecostCom}, 
demonstrate that our system is more complicated and more private, yet similarly practical to PPIS.

\begin{table}[!htbp]
\centering
\caption{
Avg. computation cost of each party per procedure over 100 runs in milliseconds.}
\label{TimecostCom}
\begin{tabular}{|c|c|c|c|c|}
\hline
Party    & Join & Earn & Exchange & Spend \\ \hline
PPIS~\cite{Bl2019UpdatableSystems}-user    & 10    & 8       & N{/A}       & 30       \\ \hline
PPIS~\cite{Bl2019UpdatableSystems}-provider & 9    & 12      & N{/A}       & 72       \\  \hline
\end{tabular}
\vspace{-0.5cm}
\end{table}

\section{Conclusion}

In this paper, we give the first study of cryptocurrency exchange that supports user anonymity and compliance requirements simultaneously. 
The platform cannot get more information from the transactions other than that has to be public. 
Users cannot get more assets from the platform so double spending is prohibited and they have to correctly report their accumulated profits for tax purposes, even in a private setting. Also, critical compliance functions are to be supported.
Our construction is efficient and achieves constant computation and communication overhead  with only simple cryptographic tools and rigorous security analysis.
Additionally, we implement our system and evaluate its practical performance.


\medskip\noindent{\bf Acknowledgement.} We would like to thank our shepherd and anonymous reviewers of NDSS24 for valuable feedbacks. This work was supported in part by research awards from Stellar Development Foundation, Ethereum Foundation, Protocol Labs, SOAR Prize, and University of Sydney’s Digital Sciences Initiative through the Pilot Research Project Scheme.


%% file: appendix.tex
\section{Related works\label{rwk}}

\noindent\textbf{Private payment.}
Payment is a basic transaction format which supports one kind of asset, and the private payment systems are built on a single private closed blockchain \cite{monero, zcash} or smart contract \cite{bunz2020zether,diamond2021manyzether}.
It does not imply the exchange between different kinds of cryptocurrencies especially for some public cryptocurrencies (like Bitcoin, Ether).

%

In the off-chain setting, many solutions have been proposed. 
They perform as opt-in tools that enhance privacy for existing cryptocurrencies.
They aim to prevent an adversary from linking a payment from a particular payer to a particular payee. 
Bolt \cite{green_bolt_2017} is an anonymous payment channel was introduced by Green and Miers.
It aims to offer privacy-preserving payment channels such that multiple payments on a single channel are unlinkable to each other. 
Assuming the funded cryptocurrency is anonymous (e.g. Zerocash), the payments in Bolt are also anonymous.

TumbleBit \cite{heilman_tumblebit_2017} is a unidirectional payment channel hub (PCH) relying on an untrusted intermediary called Tumbler and Hashed Timelock Contracts (HTLCs). 
The Tumbler issues anonymous payments that users can cash-out to Bitcoins. 
Every payment conducted through TumbleBit is backed by Bitcoins, ensuring that there is no possibility of linking individual pairs of payments. Furthermore, it is guaranteed that the Tumbler cannot engage in theft of Bitcoins, or make payments to itself.

Anonymous atomic locks (A2L) is introduced in \cite{tairia2l2021} where the authors propose a PCH upon it. 
This PCH functions as a three-party protocol designed for conditional transactions, in which an intermediary (referred to as the hub) disburses funds to the recipient contingent upon the recipient's successful resolution of a puzzle, aided by the sender.
This arrangement signifies that the sender compensates the hub. The utilization of a randomized puzzle ensures that the hub cannot establish a connection between the sender and the recipient involved in a payment.
The authors define unlinkability in terms of an interaction multi-graph \cite{heilman_tumblebit_2017}.
It is a mapping of transactions from a set of senders to a set of receivers in an epoch.
An interaction graph is called compatible if it explains the view of tumbler.
The unlinkability requires that all compatible interaction graphs are equal and the anonymity set depends on the number of compatible interaction graphs in the epoch.
Since the payment amount can be used to link the sender and receiver trivially, the unlinkability requires the amount to be fixed \cite{heilman_tumblebit_2017,tairia2l2021,glaeser2022foundations} or concealed \cite{qin2022blindhub}.

The star topology of PCH is very similar to the exchange scenario where the user sends one kind of asset to the exchange platform and receives another kind from it. 
And there are some works adding anonymity on the PCH to prevent the tumbler from linking the sender and receiver.
Regarding the exchange user as the sender and receiver at the same time, the anonymous PCH seems related to our goal that cutting the link between two accounts.
Unfortunately, it is not suitable to be used to design a private exchange system due to the model differences, operation restrictions and limited privacy:

\setlist{nolistsep}
\begin{itemize}[noitemsep,leftmargin=*]
\item[(i)] PCH requires the establishment of payment channels on the blockchains by the tumbler and users.
It means each exchange needs the deployment of two channels in two blockchains. 
The channel is only valid before the expiration time, so the establishment work should be done repeatedly to make sure that they can exchange freely.
The fund locked in the channel is fixed and the user cannot transact more than that locked amount. So the exchange amount is limited by the money locked in the channel rather than the money that the user owns.
Even if the user has a huge amount of BTC, he cannot exchange them into ETH more than the amount locked in the Ethereum channel.
\item[(ii)] The anonymity set of PCH is just the active users in an epoch.
Some constructions require the off-chain transaction amount is a fixed denomination \cite{tairia2l2021,glaeser2022foundations} which is inconvenient.
BlindHub \cite{qin2022blindhub} is a recent work that supports variable amounts.
But it still assumes that there are many active users and each of them transacts many times during the epoch.
If the sender just sends once and the receiver just receives once before closing the channel, the changed amounts of their channels would link them easily.
\item[(iii)] An exchange system consists of deposit, exchange and withdraw operations where the deposit and withdrawal amount must be public and variable.
We consider the anonymity in the whole system. The interaction graph model is not enough since it only focuses the payments in one epoch and only supports the \textit{k}-anonymity of active users. 
It drives us to define a stronger model of interaction indistinguishability with larger anonymity set.
\end{itemize}

LDSP \cite{ng_ldsp_2021} is a layer-2 cryptocurrency payment system that supports payer privacy.
It is designed in the setting of shopping with cryptocurrency where the payer is customer and the receiver is merchant.
There is also an untrusted entity called leader who is in charge of issuing coins for customers and merchants.
Customers can transfer coins off-chain with low-latency.
Since the coins are issued in a blind way, the leader cannot link the spent coin with any customer.
At the same time, the merchants are guaranteed to receive the coins.

Additionally, there exist several off-chain solutions, including Plasma \cite{Poon2017Plasma:Contracts}, NOCUST \cite{nocust}, and ZK-Rollup \cite{zkrollup}, which are designed to enhance blockchain scalability by relocating resource-intensive computations and redundant data off-chain, conducted by an untrusted operator.

To enhance privacy within these scalability-focused frameworks, the PriBank system has been introduced by Galbraith et al. \cite{galbraith_pribank_2022}. This system incorporates an efficient Commit-and-Prove Non-Interactive Zero-Knowledge (NIZK) protocol tailored for quadratic arithmetic programs. It ensures that users' balances and transaction values remain confidential, accessible only to the operator and not to other entities.

\noindent{\bf Fiat to cryptocurrency (F2C) exchange.} \ 
In general, the centralized F2C exchange platform does not consider user's privacy, like Coinbase, Binance.
They collect user's personal information when they register to meet the KYC requirement.
However, the user's accounts are transparent for the platform.
It knows their asset profile, i.e., which kinds and how many assets they own.
In the case of cryptocurrency, it would also know how the user spend their cryptocurrency which violates user's privacy outsides the platform.

To prevent the linkability by the transaction amount, the amount of withdrawn cryptocurrency is fixed for all transactions. 
For example, let all transactions worth 1 Bitcoin.
To prevent the linkability by the input UTXO, it should be chosen by the client.
But two clients may choose the same UTXO and the conflict leads to only one of them would receive the bitcoin.
Besides, it is not accountable. 
The users do not need to provide any compliance information, otherwise their privacy cannot be preserved. 

A privacy-preserving  fiat-to-Bitcoin exchange scheme is proposed in \cite{yi2019new}. 
In this scheme, a user can acquire a fixed quantity of cryptocurrency from an exchange platform using fiat currency, all the while ensuring that the platform remains unaware of the connection between the user's genuine identity and the associated Bitcoin address.
To achieve this, a blind signature mechanism is employed, allowing the user to receive Bitcoin from the platform without revealing the output address linked to the transaction. Subsequently, this transaction is recorded on the Bitcoin blockchain, divulging details such as the output address, transaction amount, and the Unspent Transaction Output (UTXO) utilized by the platform at that moment.
To mitigate the risk of linkability through transaction amounts, a constant withdrawal amount is maintained across all transactions. For example, all transactions could be set at a fixed value of 1 Bitcoin.
To counteract the potential issue of linkability through input UTXOs, clients are required to select their preferred UTXOs. However, a challenge arises when multiple clients opt for the same UTXO, potentially resulting in a conflict where only one of them receives the Bitcoin. Furthermore, this approach lacks accountability.
Crucially, users are not obligated to furnish any compliance-related information. Failure to do so would compromise their privacy preservation.


\noindent\textbf{Private decentralized exchange.}   
Decentralized exchange allows users to exchange cryptocurrencies with each other directly or with smart contract.
However, it is very different from our setting.
In the one hand, the private DEX focuses on the trade anonymity and trade confidentiality. It aims to keep the transaction information secret except for the trading parties.
But in the CEX the platform is one of the trading party who can learn the information of the other one.
On the other hand, it does not support fiat money transactions and they are generally deployed in the decentralized setting like smart contract that is unaffected by the KYC requirement.
Users are free to join the DEX without providing their real identities as long as they have cryptocurrencies.
It is hard to directly enforce compliance requirement on it since the enrollment does not require real-world identities. 

There are some works on the {private exchange} in the decentralized setting like Zexe~\cite{bowe2020zexe}, P2DEX~\cite{Baum2021P2DEXExchange} and Manta~\cite{chu_manta_nodate}, but they do not consider any compliance issue. 
P2DEX~\cite{Baum2021P2DEXExchange} is a privacy preserving exchange system for cryptocurrency tokens cross different blockchains while preserving order privacy to avoid front-running attack and ensuring users never lose tokens.
They use MPC for privately matching exchange orders and deploy smart contract to reimburse affected clients with the collateral deposit from the cheating server.
Manta \cite{chu_manta_nodate} is a decentralized anonymous exchange scheme based on automated market maker (AMM).
They design a mint mechanism to convert base coins to private coins, then achieve the decentralized anonymous exchange by trading private coins anonymously.

\noindent\textbf{Accountable privacy.} 
There are some works in studying to achieve privacy-preserving and accountability at the same time.
PGC \cite{chen2020pgc} is an auditable decentralized confidential payment system. 
It offers transaction confidentiality and two levels of auditability, namely regulation compliance and
global supervision at the same time.
Androulaki et al. \cite{Androulaki20} present a privacy-preserving token payment system for permissioned blockchains that with auditing.
The content of transactions is concealed and only some authorized parties can inspect them.

UTT \cite{tomescu2022utt} stands as a decentralized electronic cash payment system designed to incorporate accountable privacy measures. One of its key features is the integration of anonymous budgets, which contribute to maintaining a balance between privacy and accountability.
Within the UTT framework, senders are empowered to generate payments in an anonymous manner, but this is subject to a predefined monetary limit per month. Once this limit is exceeded, the system mandates that their transactions must become visible and transparent to a governing authority. This approach ensures that while users can transact with a certain degree of privacy, their financial activities remain accountable when they surpass the specified budgetary threshold.

Platypus \cite{Platypus} is a payment system designed for use within the context of a central bank digital currency (CBDC) environment. It focuses on enabling transactions that are unlinkable, ensuring privacy while also accommodating regulatory requirements. The system introduces a versatile regulatory framework, which can be applied across various scenarios, and it effectively enforces limitations on holdings and receipts as specific instances of regulatory control.

\noindent\textbf{Exchange platform compliance.} 
Provisions, as outlined in \cite{Benedikt2003ProVISION}, presents a privacy-centric approach to validating solvency within a financial exchange, particularly in the context of cryptocurrencies like Bitcoin. This scheme enables an exchange platform to demonstrate its solvency without needing to disclose sensitive information such as its Bitcoin addresses, total holdings, liabilities, or customer details.
The concept of \textit{proof of solvency} entails the exchange providing evidence that it possesses sufficient cryptocurrencies to cover each customer's account balance. This proof is composed of two primary components:
(i). \textit{proof of Liabilities}: The exchange commits to the collective quantity of Bitcoin it owes to all of its users. This commitment establishes the total liabilities of the exchange.
(ii). \textit{proof of Assets}: The exchange commits to the total value of Bitcoin over which it holds signing authority. If the value of assets under the exchange's control is equal to or greater than its total liabilities, the exchange is considered solvent.


\noindent\textbf{Privacy-preserving incentive system.}\label{IncentiveSystem}
An incentive system allows users to collect points which they can redeem later. 
Bl{\"{o}}mer et al \cite{Bl2019UpdatableSystems} proposed a privacy-preserving incentive system from an updatable anonymous credential. 
The collection and redemption are similar with the deposit and withdrawal. 
But it does not support the exchange operation and compliance regulation.
Additionally, the achieved anonymity is a weak game-based unlinkability, where the anonymity set is limited to the eligible users.

\section{cryptographic primitives}\label{AppPriliminary}
\noindent\textbf{Commitments.} 
A commitment scheme allows one to commit to a chosen value secretly, with the ability to only open to the same committed value later. 
A commitment scheme $\Pi_{\mathrm{cmt}}$ consists of the following PPT algorithms: 

\noindent$\mathsf{Setup}(1^{\lambda}) \rightarrow pp$: generates the public parameter \textit{pp}.\\
$\mathsf{Com}(m; r): \rightarrow com$ generates the commitment for the message $m$ using the randomness \textit{r}.

\noindent\textit{Hiding}. 
A commitment scheme is said to be hiding if for all PPT adversaries \adv\ and $\lambda$, it holds that

\begin{equation*}
    \left| \Pr\left[b=b' \left|
\begin{split}
& pp\leftarrow\mathsf{Setup}(1^{\lambda});\\
& (m_0, m_1)\leftarrow\adv(pp), b\sample\{0,1\},\\
& r\sample\mathcal{R}_{pp}, \textit{com} \leftarrow \mathsf{Com}(m_b;r), \\
& b'\leftarrow\adv(pp, \textit{com})
\end{split}
\right]-\frac{1}{2}
\right|\leq negl(\lambda)
\right.
\end{equation*}
If $negl(\lambda)=0$, we say this scheme is perfectly
hiding.

\noindent\textit{Binding}. 
A commitment scheme is said to be binding if for all PPT adversaries \adv\ and $\lambda$, it holds that

\begin{equation*}
\Pr\left[
\begin{split}
& \textit{com}_0=\textit{com}_1\\ 
& \wedge m_0\neq m_1    
\end{split}
\left|
\begin{split}
& pp\leftarrow\mathsf{Setup}(1^{\lambda});\\
& (m_0, m_1, r_0, r_1)\leftarrow\adv(pp),\\
& \mathsf{Com}(m_0;r_0)=\textit{com}_0,\\
& \mathsf{Com}(m_1;r_1)=\textit{com}_1
\end{split}
\right]
\leq negl(\lambda)
\right.
\end{equation*}
If $negl(\lambda)=0$, we say this scheme is perfectly
binding.

 \noindent\textbf{Blind signatures.}
A blind signature scheme  $\Pi_{\mathrm{bs}}$ for signing committed $n$ messages  has the following algorithms: 

\noindent$\mathsf{KeyGen}(pp) \rightarrow (pk, sk)$: takes public parameter $pp$ as input, outputs a key pair $(pk, sk)$. $pp,pk$ are implicit input of other algorithms for simplicity.

\noindent$\mathsf{Com}(\vec{m}, r) \rightarrow c$: given messages $\vec{m}\in\mathcal{M}^n $ and randomness $r$, computes a commitment $c$.

\noindent$\langle\mathsf{BlindSign}, \mathsf{BlindRcv}\rangle$: it is an interactive protocol between the signer and user, with inputs $(sk, c)$ and $(\vec{m}, r)$ respectively. User outputs a signature $\sigma$. 

\noindent$\mathsf{Vrfy}(\vec{m}, \sigma ) \rightarrow b$: it checks  $(\vec{m},\sigma)$ pair and outputs 0/1.

We require a blind signature scheme to be \textit{correct} and have the properties of {\em unforgeability} and {\em blindness}.

\noindent\textit{Correctness.} 
The following probability is negligible.
\begin{equation*}
\Pr\left[
\mathsf{Vrfy}(\vec{m}, \sigma ) =0
\left|
\begin{split}
& (pk, sk)\leftarrow\mathsf{KeyGen}(pp);\\
& c\leftarrow\mathsf{Com}(\vec{m}; r),\\
& \sigma\leftarrow\langle\mathsf{BlindSign}, \mathsf{BlindRcv}\rangle
\end{split}
\right]
\right.
\end{equation*}

\noindent\textit{Unforgeability.} A blind signature scheme is unforgeable if for any $ q=\mathsf{poly}(\lambda)$ and any PPT \adv\ who can query the blind signature oracle for at most $q-1$ times, the following probability is negligible.
\begin{equation*}
\Pr\left[
\begin{split}
& \forall i,j\in[q],\\
&\mathsf{Vrfy}(\vec{m}_i, \sigma_i)=1  \wedge\\ 
&\vec{m}_i\neq\vec{m}_j \pcif i\neq j
\end{split}
\left|
\begin{split}
& (pk, sk)\leftarrow\mathsf{KeyGen}(pp);\\
& \{\vec{m}_i, \sigma_i\}_{i\in[q]}\leftarrow\adv^{\mathcal{O}}(pp,pk)
\end{split}
\right]
\right.
\end{equation*}

\noindent\textit{Blindness.}A blind signature scheme is blind if for any PPT \adv\, there exists a challenger $\mathcal{C}$ who interacts with \adv\ by running $\mathsf{Com}$ and $\mathsf{BlindRcv}$ and \adv\ runs $\mathsf{BlindSign}$, the following probability is negligible. 

\begin{equation*}
    \left| \Pr\left[b=b' \left|
\begin{split}
& pp\leftarrow\mathsf{Setup}(1^{\lambda});\\
& (m_0, m_1)\leftarrow\adv(pp), b\in\{0,1\}\sample\mathcal{C}\\
& \mathcal{C}\ \mathrm{interacts\ with}\ \adv \ \mathrm{using}\ m_b, m_{1-b},\\
& \quad  \mathrm{and\ gets}\ \sigma_b, \sigma_{1-b}, \mathrm{respectively}\\
& b'\leftarrow\adv(\sigma_0,\sigma_1)
\end{split}
\right]-\frac{1}{2}
\right|
\right.
\end{equation*}

\noindent\textbf{Partially blind signature.} A partially blind signature is a variant of the blind signature, where the signed message is partially blind. Here we briefly introduce the definition and security properties following~\cite{PartialBlindSig/crypto/AbeO00}. A partially blind signature $\Pi_{\mathrm{pbs}}$ consists following three algorithms:

\noindent$\mathsf{KeyGen}(pp, 1^{n}) \rightarrow (pk, sk)$: generates a public and secret key pair $(pk, sk)$.

\noindent$\langle\mathsf{PartialBlindRcv},\mathsf{PartialBlindSign}\rangle$: it is an interactive protocol between the user and signer with inputs $(pp,pk,msg,\mathit{info})$ and $(pp,pk,sk,\mathit{info})$ respectively, where $msg$ denotes the blind part of signed message, and  $\mathit{info}$ denotes the unblind part of signed message. User outputs $\bot$ or the message signature pair $(msg,\mathit{info},\sigma)$, and signer outputs $b=0/1$ indicating whether it fails or not.

\noindent$\mathsf{Vrfy}(pp, pk, msg,\mathit{info}, \sigma ) \rightarrow b$: it checks  $(msg,\mathit{info},\sigma)$ pair and outputs 0/1.

We require a partially blind signature scheme to have the properties of completeness, unforgeability, and partial blindness as defined in~\cite{PartialBlindSig/crypto/AbeO00}.  

\noindent\textbf{Zero-knowledge argument of knowledge ($\mathsf{ZKAoK}$)\cite{lindell2003parallel}.}
A zero-knowledge argument of knowledge is a cryptographic protocol involving two participants: a prover and a verifier. In this protocol, the prover's primary aim is to convince the verifier that a specific statement is true, all while ensuring that the evidence supporting this statement, known as the witness, remains confidential. The central objective is to furnish a compelling proof without disclosing any information about the underlying witness.
This system encompasses three algorithms, namely $\textbf{Setup}$, $\mathcal{P}$, and $\mathcal{V}$, all of which run in probabilistic polynomial time. The $\textbf{Setup}$ algorithm takes a security parameter $\lambda$ as input and generates a shared reference string $\sigma$.
The prover $\mathcal{P}$ and the verifier $\mathcal{V}$ are interactive algorithms. 
The transcript produced by $\mathcal{P}$ and $\mathcal{V}$ when interacting on inputs $x$ and $y$ is denoted by $tr\gets \langle \mathcal{P},\mathcal{V} \rangle$.
As the output of this protocol, we use the notation $\langle \mathcal{P},\mathcal{V} \rangle = b$,
where $b=1$ if $\mathcal{V}$ accepts and $b=0$ if $\mathcal{V}$ rejects.

Let $\mathcal{R}$ be a polynomial-time verifiable ternary relation for common reference string $\sigma$, statement $x$, and witness $w$, and let $\mathcal{L}$ be the corresponding language, 
i.e.,
$\mathcal{L}=\{x \ |\ \exists w \text{, s.t., } (\sigma,x,w)\in \mathcal{R}  \}$.
The argument of knowledge is defined as follows.
	
\noindent\textit{Argument of Knowledge.} 
The triple $(\textnormal{\textbf{Setup}},\mathcal{P}, \mathcal{V})$ is called an argument of knowledge for the relation $\mathcal{R}$ if it satisfies the following two definitions.
	
\setlist{nolistsep}
\begin{itemize}[noitemsep,leftmargin=*]
    \item[-] \textit{Perfect completeness.}
$(\textnormal{\textbf{Setup}},\mathcal{P},\mathcal{V})$ has perfect completeness if for any $(\sigma,x,w)\in \mathcal{R}$, $\langle \mathcal{P}(\sigma,x,w),\mathcal{V}(\sigma,x) \rangle$ always outputs 1.
\item[-] \textit{Knowledge Soundness.}
$(\textnormal{\textbf{Setup}},\mathcal{P},\mathcal{V})$ has knowledge soundness with error $\kappa$ if there exists a knowledge extractor $\mathcal{E}$, s.t. for any deterministic polynomial-time prover $\mathcal{P}^*$, if $\mathcal{P}^*$ convinces $\mathcal{V}$ of $x$ with probability $\epsilon>\kappa$, then $\mathcal{E}^{\langle \mathcal{P}^*(\cdot),\mathcal{V}(\cdot)  \rangle}(x)$ outputs $w$ s.t. $(\sigma,x,w)\in \mathcal{R}$ in expected time $\frac{poly(|x|)}{\epsilon(|x|)-\kappa(|x|)}$. 
Here $\mathcal{E}$ has access to the oracle $\langle \mathcal{P}^*(\cdot),\mathcal{V}(\cdot)  \rangle$ that permits rewinding to a specific round and rerunning with  $\mathcal{V}$ using fresh randomness.
\end{itemize}

The protocols in this paper  require  the zero-knowledge property.
We define it as follows.

\noindent\textit{Zero-knowledge.} 
A public coin argument $(\textnormal{\textbf{Setup}},\mathcal{P},\mathcal{V})$ is zero-knowledge for $\mathcal{R}$ if there exists probabilistic polynomial-time simulator $S$ such that for all non-uniform polynomial-time interactive adversaries $\mathcal{A}$ and any $\lambda\in\mathbb{N}$,
		\begin{small}
			\begin{align*}
			\biggm|&\Pr\Bigg[ 
			\begin{array}{c}
			\mathcal{A}(tr)=1 \ \wedge\\
			(\sigma,x,w)\in \mathcal{R}
			\end{array}
			\biggm| 
			\begin{array}{c}
			\sigma\gets \textnormal{\textbf{Setup}}(1^\lambda);\\ (x,w,\rho)\gets \mathcal{A}(\sigma); \\
			tr\gets \langle \mathcal{P}(\sigma,x,w),\mathcal{V}(\sigma,x,\rho)  \rangle
			\end{array}
			\Bigg]\\
			-
			&\Pr\Bigg[ 
			\begin{array}{c}
			\mathcal{A}(tr)=1 \ \wedge\\
			(\sigma,x,w)\in \mathcal{R}
			\end{array}
			\biggm| 
			\begin{array}{c}
			(x,w,\rho)\gets \mathcal{A}(\sigma); \\
			tr\gets S(x,\rho)
			\end{array}
			\Bigg]\biggm|\leq negl(\lambda)
			\end{align*}
   
		\end{small}
		\!\!where $\rho$ is the randomness used by $\mathcal{V}$.

\section{Basic Withdraw Anonymity Construction\label{simple}}
In this section, we first give a construction of basic withdraw anonymity and analyze its security. Then we discuss its anonymity set which is bigger than what the defined basic withdraw anonymity could provide. So we show a simpler construction satisfying the basic withdraw anonymity and analyze the security of the simper version.

\subsection{The construction}
In the following, we construct the basic withdraw anonymity scheme $\Pi_\mathsf{BWA}$ with an additively homomorphic commitment scheme $\mathsf{Com}$, a blind signature scheme $\Pi_{\mathrm{bs}}=(\mathsf{KeyGen},\mathsf{Com},\langle\mathsf{BlindSign},\mathsf{BlindRcv}\rangle , \mathsf{Vrfy})$ using $\mathsf{Com}$ to blind messages, and a $\mathsf{ZKAoK}$ scheme.


\noindent $\mathsf{Setup}(1^\lambda)\rightarrow \textit{epp}$: 
It sets up the system parameter $epp$ that includes public parameters of all involved cryptographic primitives, some specific public parameters about the system, such as the total assets kinds $n$, the maximum balance $v_{max}=p-1$ for some super-poly $p$, and some dynamic parameters such as the current price ${pr}_i$ of each asset $i\in[n]$. For simplicity, $epp$ will be an input of all the following algorithms and protocols implicitly. 

\noindent$\mathsf{PKeyGen}(\textit{epp})\rightarrow (pk,sk)$:
The platform  generates a key pair $(pk, sk) \leftarrow \Pi_\mathrm{bs}.\mathsf{KeyGen}(\textit{epp})$, and initializes the platform internal state $st$ including but not restricted to the registered user set $\mathsf{USet}=\emptyset$, the identifier set $\mathsf{ID}=\emptyset$. For simplicity, the public key $pk$ will be an implicit input of the following algorithms of both user and the platform.


\noindent{$
\langle\mathsf{Join}(\mathit{req_{joi}}), \mathsf{Issue}(sk,st)\rangle\rightarrow (uid,\mathit{Rd_{reg}}/\bot;b,st')$}: In this procedure, a user registers into the system to get a unique id $uid$ and a registration record $\mathit{Rd_{reg}}$. Here $\mathit{Rd_{reg}}$ is an access token for authentication where $\mathit{Rd_{reg}}.cred=tk$. We do not specify the authentication method, which could be any secure one widely used in existing exchange systems. $\mathit{req_{joi}}=\textit{info}$ contains all the information for registration, especially some real identity and bank information for compliance. The platform updates the internal state $st'$ accordingly and outputs a bit $b$ indicating whether the registration succeeds. $b=1$ means the user $uid$ registered successfully, and a record will be added to $st'$ to store the state related to $uid$, such as the metadata, the following transactions, the balance, etc. The record is visible when the user logs into the system. Right when the user registers successfully, the balance is zero and finished transactions are empty.



\noindent$\langle\mathsf{Deposit}(uid,\mathit{Rd_{reg}},\mathit{req_{dep}}), \mathsf{Credit}(sk,st)\rangle\rightarrow (;b, st')$: It is a plain deposit where users deposit assets as they did in plain exchange system and the platform updates state $st'$ accordingly. Concretely, users first authenticate themselves to the system with user id $uid$ and access token $\mathit{Rd_{reg}}$, which is the login process. Then users submit deposit request $\mathit{req_{dep}}$ which includes all the information for deposit. The platform outputs $b=1$ indicating the transaction succeeds and updates the state $st'$ so that the user could see this deposit transaction, updated balance, and other related metadata.


\noindent$\langle\mathsf{Exchange}(\textit{uid},\mathit{Rd_{reg}},\mathit{req_{exc}}),\mathsf{Update}(sk,st)\rangle\rightarrow (\mathit{Rd_{ast}}; b,st')$:
This procedure includes two separate operations specified by the request attribute $\mathit{req_{exc}}.op$, where  $req_{exc}.op=pln$ indicates a plain exchange operation between different assets, and $\mathit{req_{exc}}.op=prv$ indicates a private exchange of one asset from a plain version to a private version where the amount of exchange is hidden. The private exchange is a preparation for later anonymous withdraw. Both two kinds of withdraw operations are done in the plain account, which means the user logs into the system  so that the platform knows whom it is interacting with.  The difference is in plain exchange, the platform knows the details about the transaction including the exact amounts and assets names, whereas the amount is hidden from the platform in private exchange. 


Concretely, the user logs into the system with $uid$ and $\mathit{Rd_{reg}}$ and gets the state of his account. If a user with the identifier $uid$ has previously withdrawn asset $i$, his balance state consists of two components: the plain balance denoted as $bal_i$, and a set of committed values $\{com_i^j\}_{j\in[l]}$ of withdrawn asset $i$. Here, $l$ is an integer indicating the overall count of commitments that the user has made pertaining to asset $i$. For the plain exchange, 
 the user exchanges with the request $\mathit{req_{exc}}=(i, k_i, j, k_j,pln)$. He proves that the balance of asset $i$ is greater than $k_i$. Given the balance state $\mathsf{balance}_i=(bal_i, \{com_i^j\}_{j\in[l]})$ on asset $i$, he shows that he has enough balance to exchange-out amount $k_i$ for asset $i$ by generating a proof $\pi=$ $\mathsf{ZKAoK}[\{(v_i^i,r_i^j)\}_{j\in[l]};\forall j\in[l], com_i^j=\mathsf{Com}(v_i^j;r_i^j) \wedge bal_i-k_i\geq \sum_{j\in[l]}{v_i^j}]$. Then he sends $( \mathit{req_{exc}}, \pi)$ to the platform. If the proof $\pi$ is valid, then the platform outputs $b=1$ and updates the state to $st'$, including balance state update $bal_i\leftarrow bal_i-k_i$, $bal_j\leftarrow bal_j+k_j$, and adding this exchange transaction and other metadata. Otherwise, the platform outputs $b=0$, aborts this transaction, and updates the state $st'$ with related metadata. 
 
 For the private exchange, the exchange request is $\mathit{req_{exc}}=(i, k_i, prv)$. After logging into the system, the user gets the balance state $\mathsf{balance}_i=(bal_i,\{com_i^j\}_{j\in[l]})$ about asset $i$. The user commits to the amount $k_i$ by $com_i^* \leftarrow \mathsf{Com}(k_i;r^*)$ and $\textit{c}_i\leftarrow\Pi_{\mathrm{bs}}.\mathsf{Com}(rid,i,k_i;r_1)$, and generates a proof $\pi$ for enough balance $\pi=\mathsf{ZKAoK}[(k_i,r^*,rid,r_1,\{v_i^j,r_i^j\}_{j\in[l]}); com_i^*= \mathsf{Com}(k_i;r^*)\wedge \textit{c}_i=\Pi_{\mathrm{bs}}.\mathsf{Com}(rid,i,k_i;r_1) \wedge k_i\geq 0 \wedge \forall j\in[l], com_i^j=\mathsf{Com}(v_i^j;r_i^j) \wedge bal_i-k_i\geq \sum_{j\in[l]}{v_i^j}]$. If  
 the proof $\pi$ is valid, the platform interacts with the user by running $\Pi_{\mathrm{bs}}.\langle\mathsf{BlindSign}(sk, c_i),\mathsf{BlindRcv}((rid, i, k_i), r_1)\rangle\rightarrow(b;\sigma^*/\bot)$. If $\Pi_\mathrm{bs}.\mathsf{Vrfy}(rid,i,k_i,\sigma^*)=1$ and the  platform outputs $b=1$, then the transaction succeeds. The platform updates the state $st'$ by adding the commitment $com_i^*$ to the user's balance state of asset $i$, recording the private exchange transaction and related metadata. The user obtains a new asset record $\mathit{Rd_{ast}}=(rid,i,k_i,\sigma^*)$ and can see his own updated state in the system. Otherwise, the platform outputs $b=0$, and the transaction fails. The user aborts the transaction and the platform updates the internal state accordingly.

\noindent$\langle\mathsf{Withdraw}(\mathit{req_{wit}},\mathit{Rd_{ast}})$, $\mathsf{Deduct}$ $(sk,st)\rangle\rightarrow (;b,st')$: 
This is an anonymous transaction, which means the user does not log into his account and just acts as an anonymous guest, and the platform does not know whom he is interacting with.  In the withdraw operation, the user takes the request $\mathit{req_{wit}}=(i, k_i, meta)$ and an asset record $\mathit{Rd_{ast}}$ as input, where $meta$ contains the on-chain address and other metadata possibly required for the transaction. The user sends $(\mathit{req_{wit}},\mathit{Rd_{ast}})$ directly to the platform. The platform parses $\mathit{Rd_{ast}}=(rid,i,k_i,\sigma)$. If $rid\notin \mathsf{ID}$ which means the asset record has never been used before, and $\Pi_\mathrm{bs}.\mathsf{Vrfy}(rid,i,k_i,\sigma)=1$, the platform does the on-chain transfer to the specified address in $meta$, output $b=1$ indicating the transaction succeeds, and updates the internal state accordingly including adding $rid$ to the set $\mathsf{ID}$. Otherwise, the user and platform abort the transaction.

\noindent{$\langle\mathsf{File}(uid, \mathit{Rd_{reg}}, \mathit{req_{fil}}), \mathsf{Sign}(sk,st)\rangle\rightarrow (\mathit{doc};b,st')$}: This operation is done for client compliance. In our basic construction, we follow the same client compliance rule as our full construction specified in~\ref{compliance-rule}. So the user needs to file tax for the transactions occurring in a time period. The user logs into the system with $uid$ and $\mathit{Rd_{reg}}$, and gets compliance state, his balance state, and transaction histories from the system internal state $st$. Based on the transaction histories and the state of balance, the user computes the commitment $c=\Pi_{\mathrm{bs}}.\mathsf{Com}(id, cp_1, cp_2, mt; r)$ on the user's real identity $id$, the total cost $cp_1$, the total gain $cp_2$ during the time period $mt$ and generates a proof of correct identity, total cost and gain, and time period calculation in the commitment.

According to the rule, the correct total cost and gain during $mt$ are only related to the assets the user sells by exchanging out (including plain exchange and private exchange for anonymous withdraw) during $mt$. 
The user should show the total amount $\textit{sum}_i$ of each asset $i$ he sells during $mt$ that is the sum of exchange-out in plain exchange and private exchange which is committed.  
It is easy to add the amount in plain exchange. In private exchange, the amount is committed with additive homomorphic commitment. So the commitment of the sum is the sum of each commitment. The user could open the commitment of the sum to show the total amount. Then the user needs to specify the transactions he buys them in, which occurs in deposit and plain exchange operations. With all the related deposit, plain exchange, and private exchange transaction histories, the user could calculate the gain and cost by multiplying each amount and the corresponding price and adding them. Since the price is plain value and some amount of gain is committed with  additive homomorphic commitment, the scalar multiplication and addition on the commitment could get the committed value $c_2$ of the total gain. The total cost $cp_1^*$ could be calculated with price and amount in plain.

To prove the correctness of $c$, the user generates the proof that in $c$, $cp_1$ is equal to plain calculation $cp_1^*$, $cp_2$ is the value $c_2$ commits to, $id$ is the user $uid$'s real identity (which involves the proof of same real identity), and $\mathit{mt}$ is the exact time period. 
If the user's compliance state $\mathsf{cmp}_{\mathit{mt}}=\mathsf{false}$ for the time period $\mathit{mt}$ and the proof is valid, the platform interacts with the user by running $\Pi_{\mathrm{bs}}.\langle\mathsf{BlindSign}(sk, c),\mathsf{BlindRcv}((\textit{id}, cp_1, cp_2, \mathit{mt}), r)\rangle\rightarrow(b;\sigma^*/\bot)$. If the user passes the final check, i.e., $\Pi_\mathrm{bs}.\mathsf{Vrfy}(id,i,k_i,\sigma^*)=1$  the transaction succeeds and the platform outputs $b=1$.
For the platform, it updates the internal state, including the user's balance state $\mathsf{balance}_i=(bal_i\leftarrow bal_i-sum_i,\emptyset)$ for each asset $i$, and compliance state $\mathsf{cmp}_{\mathit{mt}}\leftarrow\mathsf{true}$. For the user, 
he outputs $\mathit{doc}=(id, cp_1, cp_2, \mathit{mt}, \sigma^*)$.  Otherwise, none of the checks pass, the user aborts the transaction and the platform outputs $b=0$.

\noindent$\mathsf{Verify}(\textit{epp}, pk, \mathit{doc})\rightarrow b$: 
The authority sets the correct timestamp as $\mathit{mt}'$ from $epp$ and parses $\mathit{doc}=(id,cp_1, cp_2,\mathit{mt}, \sigma^*)$,
if $\mathit{mt}\neq \mathit{mt}'$ or $id$ is invalid (which involves some real identity check) or $\Pi_{\mathrm{bs}}.\mathsf{Vrfy}((id, cp_1,cp_2, \mathit{mt}),\sigma^*)\rightarrow 0$, 
it outputs $b=0$ indicating the verification fails. 
Otherwise, it is valid and updates $id$'s compliance state of $mt$ to $\mathsf{true}$, which is maintained by the authority.

\noindent$\mathsf{Check}(epp,st)$: $\mathsf{P}$ runs $\mathsf{Check}(epp,st)$ for self-checking the internal state's compliance with platform rules specified in $epp$. The output is a single bit $b$, with $b=1$ indicating a passing check and $b=0$ otherwise.

\noindent{\bf Anonymity set of $\Pi_\mathsf{BWA}$.} The anonymity set that $\Pi_\mathit{bs}$ provides is larger than the set provided by the security definition of basic withdraw anonymity. Because the security definition of basic withdraw anonymity only captures that for an adversary-specified amount and two eligible users (both could successfully withdraw the specified amount of coin), it is indistinguishable from which user executing the withdraw operation. This security could be achieved by purely anonymous withdraw operations with all other operations plain (or unprotected). Our construction $\Pi_\mathsf{BWA}$ includes both anonymous withdraw and private exchange via committing the transaction amount, which increases the anonymity set by adding users who privately exchange some coins with an amount different from the specified amount. For example, there are two users doing private exchange twice during a tax report year, where the user $U_1$ privately exchanges twice for 5 bitcoins, and the user $U_2$ privately exchanges once for 3 bitcoins and once for 7 bitcoins. Later, one person withdrew 5 bitcoins which is unlinkable to $U_1$ or $U_2$ for a platform using $\Pi_\mathsf{BWA}$. That is, the anonymity set is larger than the number of users who exchange exactly the same amount of coins as withdrawal.

We will present a much simpler construction with the plain deposit, plain exchange, and only anonymous withdraw later in~\ref{SimplierBWA}, where the anonymity set is exact among users exchanging the same number of coins for anonymous credentials for later withdraw.

\subsection{Security analysis of $\Pi_\mathsf{BWA}$}
We briefly analyze the security of the above basic withdraw anonymity construction $\Pi_\mathsf{BWA}$.

\begin{theorem}[Basic withdraw anonymity] Let $\Pi_{\mathrm{bs}}$ have blindness,  and $\mathsf{ZKAoK}$ be zero-knowledge, $\Pi_\mathsf{BWA}$ satisfies the basic withdraw anonymity defined in Def~\ref{defano}.
\end{theorem}
\begin{proof}[Proof sketch]
    We say $\Pi_\mathsf{BWA}$ satisfies the basic withdraw anonymity if \adv\ cannot gain any advantage to link the withdraw transaction with the user identity.  We prove this theorem according to Def~\ref{defano}. 
    
    Note in the experiment $\adv$ provides two identities with valid credentials for withdrawal respectively. 
    It means \adv\ has successfully queried the credential on the same asset with the same amount for both identities, and the commitments on identical withdrawal amounts, which stem from private exchanges, are present in both users' accounts. While \adv\ has the ability of querying the $\mathcal{O}^1_\mathsf{File}$ oracle to ascertain the total amount of commitments for asset $i$ held by both users, these commitments are determined by the private exchanges that are common to both users. Therefore, \adv\ is unable to gain any benefit from querying $\mathcal{O}^1_\mathsf{File}$ oracle in this context.
    \adv\ wins only if it can link the credential showing for withdrawal of the challenger to one of the private exchanges correctly with overwhelming probability. 
    In the private exchange phase, the view of \adv\ consists of the commitments $\textit{com}_i^*, c_i$, the proof $\pi$ and the blind signature transcript $\textit{trans}$. We define the following hybrid games:
    \setlist{nolistsep}
    \begin{itemize}[noitemsep,leftmargin=*]
        \item[-] $\mathsf{G}_0$: it is identical with the experiment in Def~\ref{defano};
        \item[-] $\mathsf{G}_1$: it is identical with $\mathsf{G}_0$ except that $\pi$ is replaced by $\pi'$ which is simulated with random strings;
        \item[-] $\mathsf{G}_2$: it is identical with $\mathsf{G}_1$ except that 
        $\mathcal{C}$ acts as the adversary on the blindness of $\Pi_{\mathrm{bs}}$: it chooses $rid_0, rid_1$ randomly w.r.t. references $\textit{ref}_0, \textit{ref}_1$ and sends $(rid_0, i, k_i), (rid_1, i, k_i)$ as the challenge to the challenger $\mathcal{B}$ in $\Pi_{\mathrm{bs}}$'s blindness experiment. Then \cdv\ interacts with $\mathcal{B}$ and forwards the transcripts \textit{trans} to \adv.
        Finally, $\mathcal{B}$ chooses a random bit $b$ and interacts with \cdv\ who blind sign $(rid_b, i, k_i)$ and $(rid_{1-b}, i, k_i)$, respectively. Then \bdv\ send two message signature pairs to $(rid_b, i, k_i,\sigma_b)$ and $(rid_{1-b}, i, k_i, \sigma_{1-b})$ to \cdv. \cdv\ forwards $(rid_0, i, k_i, \sigma_0)$ to \adv. \cdv\ forward \adv's guess $b'$ to \bdv. If $b'=b$ say \adv\ win the game by linking the withdraw operation with the credential issuance in a private exchange, which means \cdv\ could attack the blindness of underlying blind signature $\Pi_{\mathrm{bs}}$.
    \end{itemize}
    If \adv\ wins in $\mathsf{G}_2$ with non-negligible probability, then \cdv\ wins the blindness experiment of $\Pi_{\mathrm{bs}}$ also with non-negligible probability which leads to a contradiction.
    So $\adv$ wins in $\mathsf{G}_2$ with only negligible probability.
    Compared with $\mathsf{G}_1$, \cdv hits the correct challenge of \adv with the probability $1/q^2$ where $q\in \mathsf{poly}(\lambda)$ denotes the total number of \adv's queries. So there is at most $1/q^2$ reduction loss from $\mathsf{G}_1$ to $\mathsf{G}_2$.
    Since $\mathsf{ZKAoK}$ is zero-knowledge, $\mathsf{G}_1$ can be distinguished from $\mathsf{G}_0$ with only negligible probability.
    Thus $\adv$ wins in $\mathsf{G}_0$ also with only negligible probability.
\end{proof}

\begin{theorem}[Overdraft prevention] Let $\Pi_{\mathrm{bs}}$ be unforgeable, commitment be binding, and $\mathsf{ZKAoK}$ be argument-of-knowledge. $\Pi_\mathsf{BWA}$ satisfies the overdraft prevention defined in Def~\ref{od}.
\end{theorem}
\begin{proof}[Proof sketch]
    Intuitively, overdraft means \adv\ spends more than he owns. 
    Note that all transactions are plain which can be checked by the platform except the withdraw transactions. So the event that \adv\ spends more asset only happens in one of the following three cases: 
    \begin{itemize}
        \item[(1)] In the withdraw transaction, \adv\ withdraws asset with a valid credential but it has never queried the deduct oracle on it which means it is forged by itself. It violates the unforgeability of $\Pi_{\mathrm{bs}}$.
        \item[(2)] \adv\ guess other users' valid credentials which is negligible due to the randomness of credential unique identifier.
        \item[(3)] The credential is issued by the platform, but the revealed asset amount is larger than the deducted amount in the commitment or the account plain balance. In this case, \adv\ could get it by finding collisions in the commitment, which could be reduced to the binding property of commitment. \adv\ could also get it by proving a wrong statement, which is negligible due to the argument-of-knowledge property of $\mathsf{ZKAoK}$.
    \end{itemize}
    Since any of the above cases happens only with negligible probability, \adv\ also wins with negligible probability.
\end{proof}

\begin{theorem}[Tax-report-client-compliance] Let $\Pi_{\mathrm{bs}}$ be unforgeable, commitment be binding, and $\mathsf{ZKAoK}$ is an argument of knowledge. $\Pi_\mathsf{BWA}$ satisfies the Tax-report-client-compliance defined in Def~\ref{defcp} where $F$ is the tax-report function.
\end{theorem}
\begin{proof}[Proof sketch]
    Intuitively, tax-report-client-compliance requires any user to report the exact cost and gain for his account. \adv\ generates a valid \textit{doc} that contains a signature $\sigma^*$ on $(id, cp_1, cp_2, \mathit{mt})$ that can be verified by the platform's \textit{pk}. \adv\ wins if the total cost value $cp_1$ or total gain value $cp_2$ is wrong. Note that $cp_1, cp_2$ are computed from the user's transaction histories which are recorded by the platform. So, for deposit and plain exchange transactions, the platform knows the plain transaction details and could directly compute the cost and gain. For the private exchange transactions that contribute to a portion of the user' gain, the platform knows the plain price and the commitment to the amount and could calculate the commitment of the sum gain. 
    With the platform knowing the plain cost, plain gain, and committed gain, and blind-signing the compliance information based on them, \adv\ wins only in one of the following cases: 

    \begin{itemize}
        \item[(1)] The signature $\sigma^*$ of the platform was forged by \adv\ on the wrong $cp_1$ or $cp_2$. 
        If it happens, it violates the unforgeability of $\Pi_{\mathrm{bs}}$.
        \item[(2)] In the file protocol, \adv\ opens commitments recorded in his account to different values, which violates the binding property of commitment. 
        \item[(3)] In the file protocol, \adv\ proves a wrong statement on his compliance information to get less tax, which violates the argument-of-knowledge property of $\mathsf{ZKAoK}$. 
    \end{itemize}
    Since any of the above cases happens only with negligible probability, \adv\ also wins with negligible probability.
\end{proof}

\noindent{\textbf{For platform compliance.}} We observe that our basic withdraw anonymity construction does not bring any more challenges to platform compliance than existing plain exchange schemes. The reason is platform could know the exact total amount of each asset in all accounts. The transaction amount is known to the platform in deposit, plain exchange, and withdraw with one-use anonymous credential. The only case that the platform does not know the amount is the exchange preparation, while it does not exchange assets but changes the asset form from plain to anonymous version. Thus it satisfies the strictest platform compliance rule in~\cite{BinanceProofofReserve}.

\subsection{Simpler construction with basic withdraw anonymity}\label{SimplierBWA}

We construct a simpler basic withdraw anonymity scheme denoted by $\Pi_\mathsf{S-BWA}$ via the partially blind signature $\Pi_\mathrm{pbs}=(\mathsf{Keygen},$ $\langle\mathsf{PartialBlindSign},$ $\mathsf{PartialBlindRcv}\rangle, \mathsf{Vrfy})$.

The main idea to achieve basic anonymous withdraw is that before withdrawing, the user gets an anonymous credential issuance for the asset so that showing an anonymous credential in the withdraw is unlinkable to the credential issuance. The credential issuance is done as a special exchange transforming the asset form from plain asset to asset credential. $\Pi_\mathsf{S-BWA}$ is simpler than $\Pi_\mathsf{BWA}$ because its credential issuance does not protect the amount of asset, which is committed in $\Pi_\mathsf{BWA}$ and brings additional proof for the user's balance in the plain exchange and private exchange. $\Pi_\mathsf{S-BWA}$ uses the partially blind signature $\Pi_\mathrm{pbs}$ to enable credential issuance where the asset amount is public but only the random index unique for the credential is hidden from the platform.  Since the asset info is public for the user and the platform, the user does not need to prove enough balance for the following exchange. Accordingly, a withdraw operation is unlinkable to the previous special exchange operations with the same amount, which causes quite limited anonymity that $\Pi_\mathsf{S-BWA}$ could provide. But we stress $\Pi_\mathsf{S-BWA}$ is very simple and we will prove that $\Pi_\mathsf{S-BWA}$ satisfies the basic withdraw anonymity we define in~\ref{defano}.

The concrete construction is shown in the following, where for the same steps as $\Pi_\mathsf{BWA}$, we will specify it  and refer to  $\Pi_\mathsf{BWA}$'s construction description for simplicity.

\noindent $\mathsf{Setup}(1^\lambda)\rightarrow \textit{epp}$: This step is same as $\Pi_\mathsf{BWA}.\mathsf{Setup}$.
It sets up the system parameter $epp$ that includes public parameters of all involved cryptographic primitives, some specific public parameters about the system, such as the total assets kinds $n$, the maximum balance $v_{max}=p-1$ for some super-poly $p$, and some dynamic parameters such as the current price ${pr}_i$ of each asset $i\in[n]$. For simplicity, $epp$ will be an input of all the following algorithms and protocols implicitly.

\noindent$\mathsf{PKeyGen}(\textit{epp})\rightarrow (pk,sk)$:
The platform  generates a key pair $(pk, sk) \leftarrow \Pi_\mathrm{pbs}.\mathsf{KeyGen}(\textit{epp})$. The internal state $st$ initialization is the same as in $\Pi_\mathsf{BWA}.\mathsf{Setup}$. Concretely, the platform initializes the platform internal state $st$ including but not restricted to the registered user set $\mathsf{USet}=\emptyset$, the identifier set $\mathsf{ID}=\emptyset$. For simplicity, the public key $pk$ will be an implicit input of the following algorithms of both the user and the platform.

\noindent{$\langle\mathsf{Join}(\mathit{req_{joi}}), \mathsf{Issue}(sk,st)\rangle\rightarrow (uid,\mathit{Rd_{reg}}/\bot;b,st')$}: This step is the same as $\Pi_\mathsf{BWA}.\langle\mathsf{Join},
\mathsf{Issue}\rangle$. In this procedure, a user registers into the system to get a unique id $uid$ and a registration record $\mathit{Rd_{reg}}$. Here $\mathit{Rd_{reg}}$ is an access token for authentication where $\mathit{Rd_{reg}}.cred=tk$. We do not specify the authentication method, which could be any secure one widely used in existing exchange systems. $\mathit{req_{joi}}=\textit{info}$ contains all the information for registration, especially some real identity and bank information for compliance. The platform updates the internal state $st'$ accordingly and outputs a bit $b$ indicating whether the registration succeeds. $b=1$ means the user $uid$ registered successfully, and a record will be added to $st'$ to store the state related to $uid$, such as the metadata, the following transactions, the balance, etc. The record is visible when the user logs into the system. Right when the user registers successfully, the balance is zero and finished transactions are empty.

\noindent$\langle\mathsf{Deposit}(uid,\mathit{Rd_{reg}},\mathit{req_{dep}}), \mathsf{Credit}(sk,st)\rangle\rightarrow (;b, st')$: This deposit step is the same as $\Pi_\mathsf{BWA}.\langle\mathsf{Deposit},
\mathsf{Credit}\rangle$. It is a plain deposit where users deposit assets as they did in a plain exchange system and the platform updates state $st'$ accordingly. Concretely, users first authenticate themselves to the system with user id $uid$ and access token $\mathit{Rd_{reg}}$, which is the login process. Then users submit deposit request $\mathit{req_{dep}}$ which includes all the information for deposit. The platform outputs $b=1$ indicating the transaction succeeds and updates the state $st'$ so that the user can see this deposit transaction, updated balance, and other related metadata.

\noindent$\langle\mathsf{Exchange}(\textit{uid},\mathit{Rd_{reg}},\mathit{req_{exc}}),\mathsf{Update}(sk,st)\rangle\rightarrow (\mathit{Rd_{ast}};$ $b,st')$: The exchange operation is simpler than $\Pi_\mathsf{BWA}.\langle\mathsf{Exchange}, \mathsf{Update}\rangle$.
This plain exchange procedure includes two separate operations specified by the request attribute $\mathit{req_{exc}}.op$, where  $req_{exc}.op=pln$ indicates a plain exchange between different assets, and $\mathit{req_{exc}}.op=cred$ indicates a plain exchange of one asset from a number in the user account to an anonymous credential of that asset where the amount of exchange is plain. 
Both two kinds of withdraw operations are done in the plain account, which means the user logs into the system so that the platform knows whom it is interacting with. Concretely, the user logs into the system with $uid$ and $\mathit{Rd_{reg}}$ and gets the state of his account including balance and all deposit and exchange transaction histories.

When $req_{exc}.op=pln$, $\mathit{req_{exc}}=(i, k_i,j, k_j, pln)$, the user does the plain exchange with the platform as usual in any plain exchange platform. 
When $req_{exc}.op=cred$, the request $\mathit{req_{exc}}=(i, k_i, cred)$, the user sends the request $\mathit{req_{exc}}$ to the platform and randomly chooses an id $rid$.  Then they run $\Pi_\mathrm{pbs}.\langle\mathsf{PartialBlindRev},\mathsf{PartialBlindSign}\rangle$ with the inputs $(rid, \textit{info})$ and $(sk, \textit{info})$, where $rid$ is the private message of $\Pi_\mathrm{pbs}$ and $\textit{info}=(req_{exc}.name,req_{exc}.amount)$ is the public information for the user and the platform. The user's private output is $\bot$ or $(rid,\textit{info},\sigma)$ and the platform public output is $0/1$ to indicate whether the interaction fails or not. If the interaction succeeds, the user running the exchange algorithm outputs $\mathit{Rd_{reg}}=(rid,i,k_i,\sigma)$ and the platform outputs $b=1$ and update the internal state $st'$ accordingly, e.g., the user $uid$'s balance of asset $i$ is deducted by $k_i$, and the transaction is added to the history.

\noindent$\langle\mathsf{Withdraw}(\mathit{req_{wit}},\mathit{Rd_{ast}})$, $\mathsf{Deduct}$ $(sk,st)\rangle\rightarrow (;b,st')$:  The withdraw operation is the same as  $\Pi_\mathsf{BWA}.\langle\mathsf{Withdraw},
\mathsf{Deduct}\rangle$.
This is an anonymous transaction, which means the user does not log into his account and just acts as an anonymous guest, and the platform does not know whom he is interacting with.  In the withdraw operation, the user takes the request $\mathit{req_{wit}}=(i, k_i, meta)$ and an asset record $\mathit{Rd_{ast}}$ as input, where $meta$ contains the on-chain address and other metadata possibly required for the transaction. The user sends $(\mathit{req_{wit}},\mathit{Rd_{ast}})$ directly to the platform. The platform parses $\mathit{Rd_{ast}}=(rid,i,k_i,\sigma)$. If $rid\notin \mathsf{ID}$ which means the asset record has never been used before, and $\Pi_\mathrm{pbs}.\mathsf{Vrfy}(rid,(i,k_i),\sigma)=1$, the platform does the on-chain transfer to the specified address in $meta$, output $b=1$ indicating the transaction succeeds, and updates the internal state accordingly including adding $rid$ to the set $\mathsf{ID}$. Otherwise, the user and platform abort the transaction.

\noindent{$\langle\mathsf{File}(uid, \mathit{Rd_{reg}}, \mathit{req_{fil}}), \mathsf{Sign}(sk,st)\rangle\rightarrow (\mathit{doc};b,st')$}: This operation is done for client compliance, which is simpler than $\Pi_\mathsf{BWA}.\langle\mathsf{file},
\mathsf{sign}\rangle$, and as easy as filing compliance in plain exchange platform because all compliance-related information is plain for the platform. We follow the same client compliance rule as our full construction specified in~\ref{compliance-rule}. So the user needs to file tax for the transactions occurring in a time period. The user logs into the system with $uid$ and $\mathit{Rd_{reg}}$, and gets compliance state, his balance state, and transaction histories from the system internal state $st$. Based on the transaction histories and the state of balance shown in the platform, the platform could compute the user's cost $cp_1$ and profit $cp_2$ during the specified time period $mt$, where the cost $cp_1$ comes from the buying in asset in deposit and plain exchange multiplied the corresponding price, the profit $cp_2$ comes from selling out asset in plain exchange and special exchange to anonymous credentials during $mt$ multiplied the corresponding price. The mapping between the buying-in asset to the selling-out asset for tax reporting could be specified by the user or the platform depending on the rules. Here our construction does not introduce restriction, as all needed information is plain and could be dealt as in any plain exchange. Then the platform interacts with the user to partially blind sign on empty message $msg=\bot$ and the public informatio $\mathit{info}=(id, cp1,cp_2,mt)$, where $id$ is the user's real identity he showed in the registration. The user gets $doc=(id,cp_1,cp_2,mt,\sigma)$ or $\bot$, and the platform outputs $b=1/0$  and updates the internal state accordingly.

\noindent$\mathsf{Verify}(\textit{epp}, pk, \mathit{doc})\rightarrow b$: The verification operation is the same as $\Pi_\mathsf{BWA}.\mathsf{Verity}$. 
The authority sets the correct timestamp as $\mathit{mt}'$ from $epp$ and parses $\mathit{doc}=(id,cp_1, cp_2,\mathit{mt}, \sigma^*)$,
if $\mathit{mt}\neq \mathit{mt}'$ or $id$ is invalid (which involves some real identity check) or $\Pi_{\mathrm{pbs}}.\mathsf{Vrfy}(\emptyset, (id, cp_1,cp_2, \mathit{mt}),\sigma^*)\rightarrow 0$, where the message $msg=\emptyset$, it outputs $b=0$ indicating the verification fails. 
Otherwise, it is valid and updates $id$'s compliance state of $mt$ to $\mathsf{true}$, which is maintained by authority

\noindent$\mathsf{Check}(epp,st)$: The self-checking operation is the same as $\Pi_\mathsf{BWA}.\mathsf{Check}$.  $\mathsf{P}$ runs $\mathsf{Check}(epp,st)$ for self-checking the internal state's compliance with platform rules specified in $epp$. The output is a single bit $b$, with $b=1$ indicating a passing check and vice versa.

\subsection{Security analysis of $\Pi_\mathsf{S-BWA}$}
We briefly analyze the security of the above basic withdraw anonymity construction $\Pi_\mathsf{S-BWA}$.

\begin{theorem}[Basic withdraw anonymity] Let $\Pi_{\mathrm{pbs}}$ have partial blindness. $\Pi_\mathsf{S-BWA}$ satisfies the basic withdraw anonymity defined in Def~\ref{defano}.
\end{theorem}
\begin{proof}[Proof sketch]
    We say $\Pi_\mathsf{S-BWA}$ satisfies the basic withdraw anonymity if \adv\ cannot gain any advantage to link the withdraw transaction with the user identity.  We prove this theorem according to Def~\ref{defano}. 
    
    Note in the experiment $\mathrm{Exp}¬^{ano-wit}(\adv,\lambda)$, $\adv$ acts as malicious platform, \cdv\ acts as honest users. \adv\ provides two identities with valid credentials for withdrawal respectively. 
    It means \adv\ has successfully queried the credential on the same asset with the same amount for both identities. One of them chosen randomly by the challenger will show the credential in the withdraw transaction. This is the same as the blindness experiment of the partially blind signature. That means \cdv\ could act as the adversary of the blindness experiment of the partially blind signature to interact with challenger \bdv, forwarding \adv's queries and challenges to \bdv, and \bdv's responses to \adv. Then if \adv\ could link the withdraw with the exact credential issuance in exchange, then \cdv\ could the blindness game of partial blind signature, which contradicts the blindness of partially blind signature. So $\Pi_\mathsf{S-BWA}$ satisfies the basic withdraw anonymity if $\Pi_{\mathrm{pbs}}$ has blindness.  
\end{proof}

\begin{theorem}[Overdraft prevention] Let $\Pi_{\mathrm{pbs}}$ be unforgeable. $\Pi_\mathsf{S-BWA}$ satisfies the overdraft prevention defined in Def~\ref{od}.
\end{theorem}
\begin{proof}[Proof sketch]
    Intuitively, overdraft means \adv\ spends more than he owns. Note that all the transaction details in each plain transaction including deposit and exchange could be checked by the platform. So the event that \adv\ spends more assets than he owns only happens in the withdraw transaction with the following two cases: 
    
    \begin{itemize}
        \item[(1)] \adv\ forges a new credential to withdraw. 
        It violates the unforgeability of $\Pi_{\mathrm{pbs}}$.
        \item[(2)] \adv\ guess other users' valid credentials which is negligible due to the randomness of credential identifier $rid$.
    \end{itemize}
    Since any of the above cases happens only with negligible probability, \adv\ also wins with negligible probability.
\end{proof}

\begin{theorem}[Tax-report-client-compliance] Let $\Pi_{\mathrm{pbs}}$ be unforgeable. $\Pi_\mathsf{S-BWA}$ satisfies the Tax-report-client-compliance defined in Def~\ref{defcp} where $F$ is the tax-report function.
\end{theorem}
\begin{proof}[Proof sketch]
    Intuitively, tax-report-client-compliance requires any user to report the exact cost and gain for his account. \adv\ generates a valid $\textit{doc}^*$ that contains a signature $\sigma^*$ on $(id, cp_1^*, cp_2^*, \mathit{mt})$ that can be verified by the platform's \textit{pk}. \adv\ wins if the total cost value $cp_1^*$ or total gain value $cp_2^*$ is wrong. Note that in this experiment, the platform is assumed honest, $cp_1, cp_2$ are computed from the user's transaction histories which are recorded in plaintext by the platform. So,  the user should get only one valid $\textit{doc}$ with $id$ on the specific time period $mt$, where $\textit{doc}$ includes a valid partially blind signature $\sigma$ on correct public information $\mathit{info}=(id, cp_1, cp_2, \mathit{mt})$. \adv\ wins if $\textit{doc}\neq\textit{doc}^*$. The challenger \cdv\ could leverage the winning case to attack the unforgeability of the underlying signature.
    
    Since the partially blind signature scheme $\Pi_{\mathrm{pbs}}$ is unforgeable, \adv\ wins with negligible probability.
\end{proof}

\noindent{\textbf{For platform compliance.}} We observe that the construction $\Pi_\mathsf{S-BWA}$ does not bring any more challenges to platform compliance than existing plain exchange schemes. The reason is platform could know the exact total amount of each asset in all accounts. The transaction amount is known to the platform in deposit, all exchanges, and withdraw transactions.  Thus it satisfies the strictest platform compliance rule in~\cite{BinanceProofofReserve}.